\documentclass{article}

\usepackage[preprint]{neurips_2026}

\usepackage[utf8]{inputenc} % allow utf-8 input
\usepackage[T1]{fontenc}    % use 8-bit T1 fonts
\usepackage{hyperref}       % hyperlinks
\usepackage{url}            % simple URL typesetting
\usepackage{booktabs}       % professional-quality tables
\usepackage{amsfonts}       % blackboard math symbols
\usepackage{nicefrac}       % compact symbols for 1/2, etc.
\usepackage{microtype}      % microtypography
\usepackage{xcolor}         % colors
\usepackage{soul}
\usepackage{graphicx}
\usepackage{amsmath, amssymb, amsthm}
\usepackage{adjustbox}

\newtheorem{theorem}{Theorem}
\newtheorem{lemma}[theorem]{Lemma}
\newtheorem{proposition}[theorem]{Proposition}
\newtheorem{corollary}[theorem]{Corollary}

\theoremstyle{definition}

\newtheorem{assumption}[theorem]{Assumption}
\newtheorem{remark}[theorem]{Remark}

\newcommand{\E}{\mathbb{E}}

\newcommand{\R}{\mathbb{R}}

\newcommand{\dfn}{:=}
\newcommand{\dfnrev}{=:}

\newcommand{\A}{\mathcal{A}}
\newcommand{\B}{\mathcal{B}}
\newcommand{\Z}{\mathcal{Z}}
\newcommand{\W}{\mathcal{W}}
\newcommand{\I}{\mathbf{I}}
\newcommand{\bw}{\mathbf{w}}

\newcommand{\clmapx}{{\Phi}_x}
\newcommand{\clmapu}{{\Phi}_u}

\newcommand{\clmapxaff}{{\phi}_x}
\newcommand{\clmapuaff}{{\phi}_u}

\newcommand{\genDstro}{\mathcal{D}}
\newcommand{\posterior}{\rho}
\newcommand{\prior}{\pi}

\renewcommand{\Re}{\mathbb{R}}

\DeclareMathOperator{\tr}{tr}
\DeclareMathOperator{\KL}{KL}

\title{PAC-Bayesian Certificates for Quadratic Closed-Loop Control}
\author{%
  Domagoj Herceg\thanks{Fallback address: doma.herceg@gmail.com} \\
  Department of Mechanical Engineering\\
  Eindhoven University of Technology\\
  The Netherlands \\
  \texttt{d.herceg@tue.nl} \\
}

\begin{document}

\maketitle

\begin{abstract}
% PAC-Bayesian generalization bounds have found immense success in the machine learning field. However, in learning for control, they have largely been restricted to bounded or norm-type losses at best, since the natural quadratic cost yields controller-dependent Chernoff radii that collapse under broad posterior families. However, squared type losses are natural in control due to deep ties with quadratic Lyapunov forms and dissipativity. We tackle this issue by deriving a posterior-localized PAC-Bayes-Chernoff surrogate that imposes the admissibility condition on the deployed controller distribution as a whole rather than pointwise over its support. Applied to finite-horizon affine System Level Synthesis parameterizations, this yields explicit, tractable certificates for two disturbance models: exact and quadratic-upper-bound certificates under Gaussian disturbances, as well as Hoeffding-type bound under bounded noise. The resulting bounds depend on the closed-loop sensitivity map of the deployed posterior, giving a direct perturbation-robustness interpretation yet unexplored to the best of our knowledge.
% We conclude with numerical experiments to illustrate posterior localization and closed-loop sensitivity.
PAC-Bayesian bounds provide finite-sample guarantees for data-dependent randomized predictors, 
but applying them to learning-based control is difficult because the natural objective is a quadratic trajectory cost. 
Such losses are unbounded, non-Lipschitz , and lead to response-dependent Chernoff terms. 
We employ System Level Synthesis parameterization, which exposes the closed-loop trajectory map of a linear system directly 
and makes the quadratic control loss amenable to explicit certification.
Moreover, we provide a set of PAC-Bayes-Chernoff certificates for posterior distributions over feasible closed-loop responses. 
For Gaussian disturbance trajectories with arbitrary covariance, we derive an exact one-sided Gaussian transform and 
a tractable quadratic upper bound expressed through closed-loop sensitivity quantities. 
%The resulting certificate 
%is valid for all nonnegative Chernoff parameters, exploiting the one-sided upper-risk direction that is natural for deployment certification. 
We also derive %bounded-disturbance certificates and present 
a posterior-localized surrogate for settings where pointwise closed-loop response certificates are unavailable or have support related admissibility issues. 
Although PAC-Bayes certifies a non-degenerate posterior, the convex quadratic form of the SLS loss transfers the certificate to the posterior mean response. 
We present a deterministic mean response deployment result that is particularly suitable for control while retaining the stochastic posterior in the bound.
Additionally, we provide a data-driven bound for this deployment, transitioning away from an oracle bound. Minimizing this 
bound naturally results in a learning algorithm for control selection from data.
Numerical experiments on a double integrator show that the algorithm acts as a sensitivity-aware finite-sample regularizer, 
improving held-out cost and reducing closed-loop sensitivity in the low-data regime.
% Numerical experiments on a double-integrator system show that minimizing the certificate learns
% posterior distributions with lower empirical cost, lower certified risk, and substantially reduced
% closed-loop sensitivity across random seeds. Moreover, deployed mean controller shows remarkable robustness and performance in low-data regime.
\end{abstract}

% PAC-Bayesian bounds provide finite-sample guarantees for data-dependent randomized predictors, but their use in control is complicated by the native quadratic trajectory costs of linear control. These losses are unbounded, non-Lipschitz in the disturbance trajectory, and induce response-dependent Chernoff terms. We study finite-sample certification of data-selected linear controllers under these native quadratic costs. By parameterizing feasible closed-loop responses with finite-horizon System Level Synthesis, the weighted trajectory becomes an affine function of the disturbance and the control loss becomes a squared affine response. This structure enables exact one-sided Gaussian Chernoff calculations and tractable PAC-Bayes certificates expressed through closed-loop sensitivity quantities. Although the learned object is a non-degenerate posterior over feasible SLS responses, convexity of the quadratic risk transfers the certificate to the deterministic posterior mean response, avoiding the infinite-KL cost of a Dirac posterior. We further give a covariance-inflated data-driven refinement for the case where the disturbance covariance is unknown, yielding a computable deterministic certificate. Numerical experiments on a double integrator show that posterior optimization followed by mean deployment acts as a sensitivity-aware finite-sample regularizer, improving held-out cost and reducing closed-loop sensitivity in the low-data regime.

\section{Introduction}
%In the machine learning (ML) community, 
Learning methods have made great strides in the last decade, mostly due to the enormous availability of data. 
Due to the complexity of modern control problems, there has been notable interest in data-driven/learning methods in the control community~\citep{vale_data_driven_control,Dorfler_bridging,Recht2019,Markovsky2021}.
%and the control community has adopted many of them~\citep{bemporad2022recurrent}.
Unlike the usual bounded costs, such as zero-one training loss in machine learning, squared unbounded costs are the norm in control applications for the most part.
This stems from a connection to stability via Lyapunov quadratic forms~\citep{grune2016nonlinear}, the explicit solution for the canonical Linear Quadratic Regulator problem~\citep{athans2013optimal}, and the general convexity of the cost~\citep{boyd2004convex}, which enables fast and reliable solution algorithms. 
%Moreover, researchers in control usually assume more structure, hence the emergence of more specialized results.
We are often given an approximate model of the system, and the trajectory of the control actions is optimized w.r.t
the loss induced by that model. In what follows, we will interlace the terminology from both control and ML, which may result 
in slightly unfamiliar notation for both sides. However, we lean more toward the machine learning notation borrowed from PAC-Bayesian literature.

Classical linear control typically certifies a single controller, while robust control certifies against a worst-case set. In learning-based control, however, controllers are often selected from finite rollout data, and the pressing question is whether the selected controller generalizes well to unseen disturbance trajectories. As a consequence, it is important to offer guarantees that the learned controller will perform well on unseen data. 
In the Probably Approximately Correct (PAC) framework~\citep{valiant1984theory}, the goal is to provide that guarantee with high probability. In our view, a particular instance of the PAC framework, the PAC-Bayes theory\citep{friendlyPACBayes,Guedj2019_primer_OAC}, provides a natural intermediate guarantee as it certifies the expected deployment cost of a posterior distribution over hypotheses, with a KL penalty measuring how far this posterior moves from a data-independent prior. This gives a finite-sample certificate for a \emph{distribution} of hypotheses rather than a single point estimate and can be interpreted as a robustness measure regarding perturbations. Other PAC (or PAC style) methods usually found in control include conformal predictions~\citep{lindemann2024formal,anastasions_cp,vlahakis2024conformal} and the scenario approach~\citep{campi2018introduction,campi2009scenario}.

Unlike the traditional view of designing a controller directly, System Level Synthesis~\citep{wang2019_SLS} (SLS) makes this PAC-Bayesian viewpoint tractable for linear systems. By parameterizing feasible closed-loop responses directly, the weighted trajectory can be written as an affine function of the disturbance trajectory, so the finite-horizon quadratic control cost is exactly the squared norm of this response. This exposes the closed-loop sensitivity maps entering the certificate and allows explicit one-sided Chernoff calculations for Gaussian % and bounded 
disturbance models. The SLS parameterization makes this certificate remarkably interpretable: posterior uncertainty is penalized in directions that strongly amplify disturbances into quadratic state-input costs and can remain spread in directions that are benign for closed-loop performance.
Moreover, it certifies that the learned posterior over closed-loop responses generalizes from sampled disturbance trajectories to new trajectories drawn from the same deployment law. 

% The PAC-Bayes certificate does not provide worst-case robustness to arbitrary disturbances. Rather, it certifies that the learned posterior over closed-loop responses generalizes from sampled disturbance trajectories to new trajectories drawn from the same deployment law. 

\paragraph{Related work}
Despite its remarkable success in the field of machine learning, the adoption of PAC-Bayesian methods in learning control
has remained  somewhat muted. PAC-Bayesian methods have only recently been used to certify learning-based controllers by translating
finite-sample generalization guarantees from supervised learning to control. Work by~\citep{Majumdar_PACB_generalize_novel_env,majumdar2020pacbayescontrollearningpolicies}
introduced PAC-Bayes control for robotic policies that generalize across environments, with algorithms that
optimize PAC-Bayes bounds over finite or continuously parameterized policy classes. More recent work~\citep{boroujeni2025pac_new,BoroujeniPACSNOC}
has developed PAC-Bayesian optimal control frameworks for stochastic nonlinear systems, including
posterior-based controller design and stabilizing neural controller parameterizations. These works
show that PAC-Bayes is a natural tool for certifying data-dependent randomized controllers, but they
primarily focus on bounded costs. The latter works also focus on stability-preserving nonlinear control, 
assuming a known stabilizing controller, in addition to considering bounded disturbances.

Rather than targeting general nonlinear policy classes, we exploit the
finite-horizon affine SLS parameterization for linear systems, which exposes the closed-loop trajectory
as an affine function of the disturbance. This allows us to treat the native quadratic state-input cost
directly. In particular, we derive exact and tractable PAC-Bayes-Chernoff certificates for unbounded
quadratic losses under Gaussian disturbance trajectories, with arbitrary covariance.%, and bounded-noisecertificates. 
A related work to ours is a recent preprint~\citep{herceg2026distributionally} 
that combines PAC-Bayes with distributional robustness to address the mismatch between
training and deployment environments via SLS parameterization. However, they focus on norm type losses, heavily exploit the Lipschitz nature of the loss,
 and use SLS as a means to an end without explaining the structural benefits of such a choice.
%Distributionally robust PAC-Bayesian control uses Wasserstein ambiguity sets and SLS to obtain performance guarantees under distribution shift, 
%with robustness and concentration proxies tied to the closed-loop operator norm. 
In contrast, our focus is on the characterization of the true quadratic closed-loop loss under a fixed deployment
disturbance law, which was highlighted as a challenge in their work. The Gaussian certificate developed here uses the one-sided 
quadratic transform rather than a Lipschitz or sub-Gaussian proxy, yielding sensitivity terms based on the induced
closed-loop covariance. Additionally, we provide a data-driven algorithm transitioning away from oracle bounds and fully justify the use of the SLS framework.

\paragraph{Contributions.}
We make four contributions:
(i)  We show that finite-horizon SLS provides a natural setting for the \emph{stochastic} PAC-Bayes certification of quadratic closed-loop control losses.
(ii) Building on the above, we derive exact and tractable one-sided PAC-Bayes-Chernoff certificates for unbounded quadratic closed-loop losses under Gaussian disturbance trajectories with arbitrary covariance. 
(iii) We also prove a deterministic mean-response certificate. The non-degenerate posterior makes the complexity term well-defined, while the quadratic SLS identity transfers the randomized certificate to the deployed posterior mean, up to a curvature mismatch term (empirical to population mismatch), blending deterministic deployment with stochastic certification.
(iv) We derive a data-driven certificate by inflating the empirical covariance. Hence, we turn an oracle bound into a computable one.
Additionally, we show empirically on a double integrator that posterior optimization followed by mean deployment outperforms 
point optimization in the low to mid data regime, yielding lower held-out cost and lower closed-loop sensitivity.

\section{PAC-Bayes Setup and Quadratic SLS Losses}
In this section we introduce the necessary preliminaries regarding PAC-Bayes theory and the SLS parameterization.
Moreover, we discuss how the SLS parametrization naturally blends with PAC-Bayes, a connection not made before in the literature.

\subsection{PAC-Bayes}
\label{sec:pac-bayes-prelims}
PAC-Bayesian methods~\citep{Catoni2007,mcallester1998some} have garnered a lot of attention due to their powerful ability to 
provide generalization guarantees for \emph{randomized} predictors, with the most notable 
showcase being the first non-vacuous bounds for deep neural networks~\citep{DziugaiteR17,perez2021tighter}.
A typical PAC-Bayesian inequality aims  to bound the \emph{population risk}, 
the performance after deployment, by a sum of expected 
empirical performance and a complexity term. 
%An easy to read introduction to PAC-Bayes 
%theory are recent overview articles~\cite{friendlyPACBayes,Guedj2019_primer_OAC}. 
Let us denote the space of all hypothesis/control policies/closed-loop responses with $\Theta \subseteq \Re^d$. 
We will emphasize this overloaded nomenclature often throughout the paper to highlight the bridge 
between the disciplines.
Our goal is to infer 
a \emph{distribution} $\posterior \in \mathcal{M}_1(\Theta)$, called the \emph{posterior} distribution. 
We briefly note that the posterior need not be the Bayesian posterior, see~\citep{germain2016pac_meets_bayesian}.
This \emph{inference} is based on a dataset  $S  = \{w_i\}_{i=1}^n$ where the data is assumed to be i.i.d. 
sampled from an unknown data-generating distribution $\genDstro$ supported on $\W\subseteq \Re^w$. % Let $\bar \Re_+ \dfn \Re_+ \cup \infty$,
Then the abstract loss function can be stated as
\begin{align}
\ell: \Theta \times \W \mapsto \Re_{+}.
\label{eq:loss_abstract}
\end{align}
We define the population risk and the empirical risk (computable from data) as
\begin{align} 
L(\theta) \dfn \E_{w \sim \genDstro} \left[ \ell(\theta, w)\right]
\qquad 
\widehat L(\theta) \dfn \frac{1}{n} \sum_{i=1}^n \ell(\theta,w_i).
\label{eq:population_and_empricial_risk}
\end{align}
PAC-Bayes inequality then bounds the difference between the two after taking the expectation over 
the posterior $\theta \sim \posterior$. To ease the notation, we will often write $\E_\posterior[ L(\theta)]$ 
and $\E_\posterior [\widehat L(\theta)]$ with the implicit understanding that we are averaging over the variable in $\Theta$, or  $\theta \sim \rho$ in the above. 
In addition, we need to specify a \emph{data-independent prior} distribution $\prior \in \mathcal{M}_1(\Theta)$ 
that will act as an anchor against overfitting by penalizing the deviation of the posterior w.r.t the prior.
Finally, we specify the confidence  $\delta \in (0,1]$, indicating that the bound holds with a probability of no less than 
$1-\delta$. 
\paragraph{PAC-Bayes Chernoff bounds}
There are different types of PAC-Bayes bounds, Langford-Seeger-Maurer type~\citep{langford2001bounds,maurer2004note}, 
McAllester style~\cite{mcallester2003pac}, and the literature on the subject is vast. We focus on typically 
more modern results for unbounded costs and, more specifically, on the results in~\citep{PAC-Bayes-Chernoff-Bounds-Unbounded-Losses}.
The core idea of their approach is to bound the cumulant generating function (CGF) by a function $\psi$.
\begin{assumption}[$\theta$-dependent bounded CGF]
\label{ass:bounded_cgf}
For every $\theta\in\Theta$, define the centered CGF
\begin{equation}
\Lambda_\theta(\lambda)
:=
\log \E_{w \sim \genDstro} \left[
\exp\!\left(
\lambda\bigl(L(\theta)-\ell(\theta,w)\bigr)
\right)\right].
\label{eq:centered_cgf}
\end{equation}
There exist $b>0$ and a nonnegative convex function  $\psi(\theta,\lambda)$ such that $\psi(\theta,0) =\psi'(\theta,0) = 0$ and
\begin{equation}
\Lambda_\theta(\lambda)\le \psi(\theta,\lambda),
\qquad
\forall \lambda\in[0,b).
\label{eq:cgf_domination}
\end{equation}
\end{assumption}

Assumption~\ref{ass:bounded_cgf} is essential for the derivation of the following theorem.
\begin{theorem}[PAC-Bayes bound under bounded CGF~\citep{PAC-Bayes-Chernoff-Bounds-Unbounded-Losses}]
\label{thm:pb_cgf}
Suppose Assumption~\ref{ass:bounded_cgf} holds. Then for any $\delta\in(0,1)$, 
with probability at least $1-\delta$ over the draw of $S\sim\mathcal{D}^n$, 
the following holds simultaneously for all posteriors $\posterior\ll\pi$:
\begin{equation}
\E_{\theta\sim\posterior}[L(\theta)]
\le
\E_{\theta\sim\posterior}[\widehat{L}(\theta)]
+
\underbrace{
\inf_{\lambda \in [0,b)}
\left\{
\frac{\KL(\posterior \| \prior)+\log(n/\delta)}{\lambda(n-1)}
+
\frac{\E_{\theta\sim\posterior}[\psi(\theta,\lambda)]}{\lambda}
\right\}}_{\mathrm{Comp}(\rho,S,\delta)},
\label{eq:pb_cgf_main}
\end{equation}
\end{theorem}

There are a few things to note here that distinguish \eqref{eq:pb_cgf_main} from a more typical PAC Bayes bound of this type. 
First, the minimization in $\lambda$ is exact and avoids any clumsy union bound arguments over a predetermined grid. 
%The price to pay upfront is the logarithmic term $\log(n)$ hence, it can be considered cheap. 
Even more crucially, this bound avoids taking the expectation of the exponential of $\psi$ 
under the prior and directly uses the posterior $\posterior$ in the term $\E_{\theta\sim\posterior}[\psi(\theta,\lambda)]$. 
This is particularly well suited for control, as the cost will depend on the hypothesis/controller/closed-map distribution
and not on the prior. Hence, we obtain a bound that depends on the performance of the actual optimized posterior we deploy.
The term under the $\inf$ operator is the complexity term in which we can recognize $\KL(\posterior \| \prior)$, 
an information-centric measure that punishes the information gain of the posterior with respect to the prior. 

% \paragraph{Drawback for typical control setting of squared loss function}
% However, there is a drawback to this bound for our use case in control when considering the natural setting of squared loss for
% the state-input pair. 
% The main hurdle is that the $\lambda$ interval of validity $[0,b)$ would typically be model dependent as well, i.e., 
% $[0, b(\theta) )$. Unfortunately, in a typical unconstrained parameterization, $b(\theta)$ 
% would shrink to zero as it necessarily encompasses the range where the posterior puts any nonnegative probability. 
% To avoid this issue, we develop a novel bound that sits between the 
% one proposed in~\cite{PAC-Bayes-Chernoff-Bounds-Unbounded-Losses} and the "old-style" that averages $\Lambda$ 
% under the prior. We aim to offer a bound made for control that does not avoid exponential averaging but
% does so under the posterior, retaining the crucial property that losses should depend on the deployed controller/hypothesis.

\subsection{System Level Synthesis}
\label{sec:sls-prelims}
To enable a tractable closed-loop and learning-based analysis, we adopt a finite-horizon 
System Level Synthesis (SLS)~\citep{wang2019_SLS} parameterization for control. The key idea is to directly 
parameterize the closed-loop response from disturbances to state and input trajectories 
instead of the standard mapping through control policies. The original SLS formulation was 
developed for linear state feedback policies, which were later extended to affine policies~\citep{SLS_Tube_MPC_Zeilinger,schuepp2025_SLS_affine}.
%Here, we adopt notation from~\citep{schuepp2025_SLS_affine}.
The motivation was the well known fact that quadratic  Model Predictive Control (MPC)~\citep{rawlings2020model}
with polyhedral constraints admits an affine 
feedback law depending on the region of the state space~\citep{bemporad2002explicit}.
For more details about SLS parametrization, see Appendix~\ref{app:sls}.

\paragraph{Linear dynamical system}
Consider a linear time invariant (LTI) system over a finite horizon
\begin{equation}
\bar{x}_{t+1} = A \bar{x}_t + B \bar{u}_t + \bar{w}_t, \qquad t=0,\dots,T-1.
\label{eq:LTI_basic}
\end{equation}
with $\bar{x}_t\in\mathbb R^{n_x}$, $\bar{u}_t\in\mathbb R^{n_u}$, $\bar{w}_t\in\mathbb R^{n_x}$.
The standard control approach  would be to parameterize the controller as state (static for simplicity) 
feedback $\bar{u}_t = K\bar{x}_t$ , $K \in \mathbb R^{n_u \times n_x}$, and rewrite the dynamics as
$
\bar{x}_{t+1} = (A  + BK) \bar{x}_t + \bar{w}_t. %\qquad t=0,\dots,T-1.
%label{eq:LTI_controller}
$
However, this makes the bound less explicit. Even though it seems natural, this 
introduces a non-convex landscape concerning system design, as trajectories depend on polynomials of $(A + BK)$.
\paragraph{SLS parametrization for the stacked dynamics}
The evolution of system~\eqref{eq:LTI_basic} can be written more compactly using stacked notation. 
Define the stacked vectors 
$$
x :=
\begin{bmatrix}
\bar{x}_0 \; x_1 \; \hdots \; \bar{x}_T
\end{bmatrix}^\top, 
\quad
u :=
\begin{bmatrix}
\bar{u}_0 \; u_1 \; \cdots \; \bar{u}_{T-1}
\end{bmatrix}^\top,
\quad
w :=
\begin{bmatrix}
\bar{x}_0 \; \bar{w}_0 \; \cdots\; \bar{w}_{T-1}
\end{bmatrix}^\top.
$$
of appropriate dimensions.
% With the matrix
% $$
% \mathbf{K} = 
% \begin{bmatrix}
% K^{0,0} & 0 & 0 & \cdots & 0 \\
% K^{1,1} & K^{1,0} & 0 & \cdots & 0\\
% \vdots & \vdots & \ddots & \ddots & 0 \\
% K^{T,T} & K^{T,T-1} & K^{T,T-2} & \cdots & K^{T,0}\\
% \end{bmatrix}.
% $$
% define a linear time-varying control law. 
Each sample $w_i$ is a finite-horizon disturbance trajectory. We note that the coordinates of a trajectory need not be independent.
Then, the system \emph{closed-loop} maps from disturbance $w$ to 
the state and input trajectories can be written as
\begin{equation}
\begin{bmatrix}
x \\ u
\end{bmatrix}
=
\begin{bmatrix}
\clmapx \\
\clmapu \\
\end{bmatrix}
w\\
+
\begin{bmatrix}
\clmapxaff \\
\clmapuaff \\
\end{bmatrix}
=
\begin{bmatrix}
\clmapx & \clmapxaff \\
\clmapu & \clmapuaff \\
\end{bmatrix}
\begin{bmatrix}
w\\
1\\
\end{bmatrix}
\label{eq:sls_affine}
\end{equation}
where $\clmapx\in\mathbb R^{(T+1)n_x\times (T+1)n_x}$ and $\clmapu\in\mathbb R^{Tn_u\times (T+1)n_x}$ are block lower-triangular operators encoding causality. 
Affine terms are $\clmapxaff\in\mathbb R^{(T+1)n_x},\clmapuaff\in\mathbb R^{Tn_u}$.
%Setting $\clmapxaff, \clmapuaff$ to zero matrices recovers the linear parameterization.
Instead of solving for the \emph{control policies} , we are interested in optimizing directly 
over the \emph{closed-loop} response map represented by a tuple $\theta \dfn \{\clmapx, \clmapu, \clmapxaff,\clmapuaff\}$.
% which now becomes an optimization variable. 
However, response maps cannot be freely chosen as desired.
Any closed-loop response tuple must obey the dynamics imposed by~\eqref{eq:LTI_basic}. 
More precisely, the closed-loop response tuple must be consistent with the linear \emph{achievability constraint}
\begin{align}
\begin{bmatrix}
(\I-\Z\A) & - \Z\B
%I-\mathcal Z_A & -\mathcal Z_B
\end{bmatrix}
\begin{bmatrix}
\clmapx & \clmapxaff \\
\clmapu & \clmapuaff \\
\end{bmatrix}
=
\begin{bmatrix}
\I & d 
\end{bmatrix},
\label{eq:sls_affine_constraint}
\end{align}
where $\mathcal{Z}$ is the block-downshift operator, and $\mathcal{A},\mathcal{B}$ denotes the block-diagonal lifted system matrices. In our setting
$d = 0$, but in general, it need not be so.
What is important is that \eqref{eq:sls_affine_constraint} is a linear constraint in the tuple, and it parameterizes all closed-loop responses. 
Moreover, the linear part of a \emph{time-varying} control policy can be recovered from the closed loop responses, see Appendix~\ref{app:sls} for 
further explanations and definitions of $\Z,\mathcal{A}$ and $\mathcal{B}$.
\subsection{System Level Synthesis for Stochastic Quadratic Control}
In this subsection, we take a step further than the standard literature and describe how 
SLS is the natural coordinate system for stochastic deployment in the native quadratic cost setting.
To the best of our knowledge, we are the first to do so.

\paragraph{Nullspace coordinates and feasible response directions.}
Vectorize the response tuple as
\begin{equation}
    \vartheta
    :=
    \begin{bmatrix}
        \operatorname{vec}(\Phi_x)\\
        \operatorname{vec}(\Phi_u)\\
        \varphi_x\\
        \varphi_u
    \end{bmatrix}.
\end{equation}
After incorporating the causal sparsity constraints, the achievability
condition \eqref{eq:sls_affine_constraint} can be written as a finite
dimensional linear system $H\vartheta=h$. This was already noted in~\citep{herceg2026distributionally} and the authors stopped there 
as their goal was to have a basis that can be used with posteriors that have unbounded support.
Here, instead, we dive deeper into the particular properties of the quadratic cost and offer a much more 
complete characterization.
Let $\vartheta_0$ be one feasible solution of
$H\vartheta=h$, and let
\begin{equation}
    N=
    \begin{bmatrix}
        N_1 & N_2 & \cdots & N_p
    \end{bmatrix}
\end{equation}
be a fixed basis for $\ker(H)$. Then every feasible causal affine
response can be written as
\begin{equation}
    \vartheta(\alpha)
    =
    \vartheta_0+N\alpha
    =
    \vartheta_0+\sum_{k=1}^p\alpha_kN_k,
    \qquad
    \alpha\in\mathbb{R}^p.
    \label{eq:nullspace-sls}
\end{equation}
Indeed,
$
    H\vartheta(\alpha)
    =
    H\vartheta_0+\sum_{k=1}^p\alpha_kHN_k
    =
    h.
$
Thus, every $\alpha\in\mathbb{R}^p$ corresponds to a feasible causal
closed-loop response. We place the PAC-Bayes prior and posterior over
these free coordinates. The basis $N$ fixes the coordinate system in
which prior and posterior complexity are measured and is chosen before
observing the data. Further, each nullspace direction $N_k$ can be partitioned as
\begin{equation}
    N_k
    =
    \begin{bmatrix}
        \operatorname{vec}\!\left(\Phi_x^{[k]}\right)\\
        \operatorname{vec}\!\left(\Phi_u^{[k]}\right)\\
        \varphi_x^{[k]}\\
        \varphi_u^{[k]}
    \end{bmatrix},
    \qquad
    k=1,\ldots,p,
    \label{eq:nullspace-direction-partition}
\end{equation}
while $\vartheta_0$ determines the corresponding objects
$\Phi_x^{[0]},\Phi_u^{[0]},\varphi_x^{[0]},\varphi_u^{[0]}$ and
\begin{equation}
\begin{aligned}
    \Phi_x(\alpha)
    &=
    \Phi_x^{[0]}
    +
    \sum_{k=1}^p\alpha_k\Phi_x^{[k]},
    &
    \Phi_u(\alpha)
    &=
    \Phi_u^{[0]}
    +
    \sum_{k=1}^p\alpha_k\Phi_u^{[k]},
    \\
    \varphi_x(\alpha)
    &=
    \varphi_x^{[0]}
    +
    \sum_{k=1}^p\alpha_k\varphi_x^{[k]},
    &
    \varphi_u(\alpha)
    &=
    \varphi_u^{[0]}
    +
    \sum_{k=1}^p\alpha_k\varphi_u^{[k]}.
    \label{eq:affine-response-components}
\end{aligned}
\end{equation}

\paragraph{Weighted response maps and the two affine structures.}
Let
$$
    \bar{Q}
    :=
    \operatorname{blkdiag}(Q_0,\ldots,P),
    \qquad
    \bar{R}
    :=
    \operatorname{blkdiag}(R_0,\ldots,R_{T-1}),
$$
where $P$ may be the terminal weight. Define the weighted
closed-loop trajectory
\begin{equation}
    y_\alpha(w)
    :=
    \begin{bmatrix}
        \bar{Q}^{1/2}x\\
        \bar{R}^{1/2}u
    \end{bmatrix}
    =
    M(\alpha)w+m(\alpha),
    \label{eq:weighted-response}
\end{equation}
where
\begin{equation}
    M(\alpha)
    :=
    \begin{bmatrix}
        \bar{Q}^{1/2}\Phi_x(\alpha)\\
        \bar{R}^{1/2}\Phi_u(\alpha)
    \end{bmatrix},
    \qquad
    m(\alpha)
    :=
    \begin{bmatrix}
        \bar{Q}^{1/2}\varphi_x(\alpha)\\
        \bar{R}^{1/2}\varphi_u(\alpha)
    \end{bmatrix}.
    \label{eq:weighted-response-maps}
\end{equation}
The quadratic finite-horizon control loss is therefore
\begin{equation}
    \ell(\alpha,w)
    :=
    \|y_\alpha(w)\|_2^2
    =
    \|M(\alpha)w+m(\alpha)\|_2^2.
    \label{eq:quadratic-sls-loss}
\end{equation}

By \eqref{eq:affine-response-components}, the weighted response maps
are affine in $\alpha$:
\begin{equation}
    M(\alpha)
    =
    M_0+\sum_{k=1}^p\alpha_kM_k,
    \qquad
    m(\alpha)
    =
    m_0+\sum_{k=1}^p\alpha_km_k,
    \label{eq:affine-cost-maps}
\end{equation}
where
\begin{equation}
    M_k
    :=
    \begin{bmatrix}
        \bar{Q}^{1/2}\Phi_x^{[k]}\\
        \bar{R}^{1/2}\Phi_u^{[k]}
    \end{bmatrix},
    \qquad
    m_k
    :=
    \begin{bmatrix}
        \bar{Q}^{1/2}\varphi_x^{[k]}\\
        \bar{R}^{1/2}\varphi_u^{[k]}
    \end{bmatrix},
    \qquad
    k=0,\ldots,p.
    \label{eq:weighted-nullspace-directions}
\end{equation}
The pair $(M_k,m_k)$ has a direct closed-loop interpretation as it is 
the cost weighted change in the disturbance induced map in the $k$-th feasible nullspace direction
$N_k$. 
% In particular, for fixed $w$,
% $$
%     M_kw+m_k
% $$
% is the change in the weighted state-input trajectory generated by
% changing $\alpha_k$ by one unit. Thus, the matrices $M_k$ are not
% arbitrary features; they are dynamically feasible response directions
% seen through the quadratic control cost.
The representation above gives two complementary affine structures.
First, for fixed $\alpha$, the map
$$
    w\mapsto M(\alpha)w+m(\alpha)
$$
is affine in the disturbance trajectory. Therefore,
$\ell(\alpha,w)$ is quadratic in $w$. Under Gaussian disturbances,
this is precisely the structure used in Section~\ref{sec:gaussian_cert} to
derive an exact one-sided Gaussian Chernoff transform.
Second, for fixed $w$, \eqref{eq:affine-cost-maps} gives
\begin{equation}
\begin{aligned}
    M(\alpha)w+m(\alpha)
    =
    M_0w+m_0
    +
    \sum_{k=1}^p\alpha_k(M_kw+m_k) =
    a_0(w)+\mathcal A(w)\alpha,
    \label{eq:affine-alpha-trajectory}
\end{aligned}
\end{equation}
where
\begin{equation}
    a_0(w):=M_0w+m_0,
    \qquad
    \mathcal A(w)
    :=
    \begin{bmatrix}
        M_1w+m_1 &
        \cdots &
        M_pw+m_p
    \end{bmatrix}.
    \label{eq:alpha-feature-map}
\end{equation}
Consequently,
\begin{equation}
    \ell(\alpha,w)
    =
    \|a_0(w)+A(w)\alpha\|_2^2
    \label{eq:quadratic-alpha-loss}
\end{equation}
is convex quadratic in $\alpha$. Hence, both the population and
empirical risks are convex quadratic functions of the SLS coordinate.
This second structure is used in Section~\ref{sec:mean-response-deployment} to
transfer a randomized PAC-Bayes certificate to the deterministic mean
response. 

\paragraph{Takeaway} In summary, SLS is central to our approach because it simultaneously guarantees multiple desirable properties:
(i) the nullspace
parameterization guarantees closed-loop feasibility for posterior 
distributions with unbounded support,
(ii) it makes the loss quadratic in the disturbance trajectory for each fixed response, which enables the certificate in Section~\ref{sec:gaussian_cert}, 
(iii) it makes the loss quadratic in the
free response coordinate for each fixed disturbance realization, which enables deterministic deployment in Section~\ref{sec:mean-response-deployment}
, and finally 
(iv) the feasible basis directions $(M_k,m_k)$ describe how posterior uncertainty over response
coordinates maps the disturbances to cost, which is used in the data-driven certificate to quantify the
additional covariance-mismatch penalty (see Section~\ref{sec:mean-response-deployment}).

\section{Posterior-localized PAC-Bayes bounds}
Theorem~\ref{thm:pb_cgf} is directly useful when one has a pointwise certificate
$
\Lambda_\alpha(\lambda)\le \psi(\alpha,\lambda)
$
valid on a common interval $[0,b)$. 
If the posterior distribution has unbounded support, the common admissible range over the posterior support may
collapse even when most posterior mass is concentrated on "good" responses.
For sub-Gaussian disturbances, Hanson-Wright~\citep{hanson1971bound} yields such a
certificate but only on a posterior-dependent interval
$\lambda \in [0,\,(C\sigma^2\|M(\alpha)\|_\mathrm{op}^2)^{-1})$,
 which may collapse to zero under broad unbounded posteriors whenever the posterior assigns mass
to directions along which $\|M(\alpha)\|_{\rm op}$ is unbounded. Corollary~\ref{thm:posterior-localized-pb-bound} 
targets exactly this setting.

\subsection{A posterior-localized Chernoff surrogate}
Instead of first certifying each
response and then averaging the certificate, we upper-bound the posterior average of the 
response-wise
CGF by a single joint exponential moment under the deployed posterior.

\begin{corollary}[Posterior-localized Chernoff surrogate]
\label{thm:posterior-localized-pb-bound} 
Let $S=(w_1,\ldots,w_n)\sim\mathcal D^n$ with $n\geq 2$, and suppose
that $\ell(\alpha,w)\geq 0$ and $L(\alpha)<\infty$ for every
$\alpha\in\Theta$. Define
$$
\widehat L_S(\alpha)
:=
\frac{1}{n}\sum_{i=1}^n \ell(\alpha,w_i),
$$
and, for $\lambda\geq 0$,
$$
\Lambda_\alpha(\lambda)
:=
\log
\mathbb E_{w\sim\mathcal D}
\left[
\exp\left(
\lambda\bigl(L(\alpha)-\ell(\alpha,w)\bigr)
\right)
\right].
$$
For a posterior $\rho\ll\pi$, define the joint posterior-localized CGF
$$
\widehat\Lambda_\rho(\lambda)
:=
\log
\mathbb E_{\alpha\sim\rho,\;w\sim\mathcal D}
\left[
\exp\left(
\lambda\bigl(L(\alpha)-\ell(\alpha,w)\bigr)
\right)
\right],
$$
and its effective domain
$$
\widehat I_\rho
:=
\left\{
\lambda>0:
\widehat\Lambda_\rho(\lambda)<\infty
\right\}.
$$
Then, with a probability of at least $1-\delta$ over $S\sim\mathcal D^n$,
simultaneously for all $\rho\ll\pi$ with
$\widehat I_\rho\neq\emptyset$,
$$
\mathbb E_{\alpha\sim\rho}[L(\alpha)]
\leq
\mathbb E_{\alpha\sim\rho}[\widehat L_S(\alpha)]
+
\inf_{\lambda\in\widehat I_\rho}
\left\{
\frac{
\operatorname{KL}(\rho\|\pi)+\log(n/\delta)
}{
(n-1)\lambda
}
+
\frac{
\widehat\Lambda_\rho(\lambda)
}{
\lambda
}
\right\}.
$$
\end{corollary}

\begin{proof}
Since $\ell(\alpha,w)\geq 0$, for every $\alpha\in\Theta$ and
$\lambda\geq 0$,
$$
\Lambda_\alpha(\lambda)
\leq
\log
\mathbb E_{w\sim\mathcal D}
\left[
\exp\bigl(\lambda L(\alpha)\bigr)
\right]
=
\lambda L(\alpha)
<
\infty.
$$
Moreover, $\Lambda_\alpha$ is convex, nonnegative, and satisfies
$\Lambda_\alpha(0)=\Lambda_\alpha'(0)=0$. Hence, Assumption~\ref{ass:bounded_cgf} holds
with $b=\infty$ by trivially choosing 
$
\psi(\alpha,\lambda)=\Lambda_\alpha(\lambda).
$
Applying Theorem~\ref{thm:pb_cgf} yields, with probability at least $1-\delta$,
simultaneously for all $\rho\ll\pi$,
$$
\mathbb E_{\alpha\sim\rho}[L(\alpha)]
\leq
\mathbb E_{\alpha\sim\rho}[\widehat L_S(\alpha)]
+
\inf_{\lambda>0}
\left\{
\frac{
\operatorname{KL}(\rho\|\pi)+\log(n/\delta)
}{
(n-1)\lambda
}
+
\frac{
\mathbb E_{\alpha\sim\rho}[\Lambda_\alpha(\lambda)]
}{
\lambda
}
\right\}.
$$

Now fix $\lambda\in\widehat I_\rho$. By the definitions of
$\Lambda_\alpha$ and $\widehat\Lambda_\rho$, and using Tonelli~\citep{schilling2017measures} to exchange 
expectations, we obtain
$$
\begin{aligned}
\mathbb E_{\alpha\sim\rho}
\left[
\exp\bigl(\Lambda_\alpha(\lambda)\bigr)
\right]
&=
\mathbb E_{\alpha\sim\rho,\;w\sim\mathcal D}
\left[
\exp\left(
\lambda\bigl(L(\alpha)-\ell(\alpha,w)\bigr)
\right)
\right] = 
\exp\bigl(\widehat\Lambda_\rho(\lambda)\bigr)
<
\infty.
\end{aligned}
$$
Therefore, Jensen's inequality gives
$$
\mathbb E_{\alpha\sim\rho}
\left[
\Lambda_\alpha(\lambda)
\right]
\leq
\log
\mathbb E_{\alpha\sim\rho}
\left[
\exp\bigl(\Lambda_\alpha(\lambda)\bigr)
\right]
=
\widehat\Lambda_\rho(\lambda).
$$
Combining this inequality with the bound from Theorem~\ref{thm:pb_cgf} gives, for every
$\lambda\in\widehat I_\rho$,
$$
\mathbb E_{\alpha\sim\rho}[L(\alpha)]
\leq
\mathbb E_{\alpha\sim\rho}[\widehat L_S(\alpha)]
+
\frac{
\operatorname{KL}(\rho\|\pi)+\log(n/\delta)
}{
(n-1)\lambda
}
+
\frac{
\widehat\Lambda_\rho(\lambda)
}{
\lambda
}.
$$
Taking the infimum over $\lambda\in\widehat I_\rho$ proves the result.
\end{proof}

We call the bound \emph{posterior-localized} because admissibility is imposed on the joint law
$(\alpha,w)\sim\rho\times \genDstro$, rather than pointwise over every response in the posterior support.
When explicit pointwise certificates are available, the averaged-CGF route in
Theorem~\ref{thm:pb_cgf} is sharper. The role of Corollary~\ref{thm:posterior-localized-pb-bound}
is to provide a fallback when response certificates are unavailable or have support 
admissibility issues.

\begin{remark}[Posterior-level admissibility]
The admissibility requirement in Corollary~\ref{thm:posterior-localized-pb-bound}
is posterior-localized. In particular, it does not require a common
Chernoff domain over the entire feasible SLS class. Rather, the bound is
meaningful for those posteriors $\rho$ for which
$
\widehat I_\rho
$
% $=
% \left\{
% \lambda>0:
% \log
% \mathbb E_{\alpha\sim\rho,w\sim \genDstro}
% \exp\{\lambda(L(\alpha)-\ell(\alpha,w))\}
% <\infty
% \right\}
% $$
is nonempty. The feasible SLS class may contain responses with poor
behavior, while the learned posterior may concentrate on a particularly desirable 
neighborhood of closed-loop responses. The joint posterior-CGFs
form is naturally conservative relative to the posterior average of pointwise CGFs since it uses
Jensen's inequality.

% gives
% $$
% \mathbb E_{\alpha\sim\rho}
% \left[
% \log
% \mathbb E_{w\sim \genDstro}
% \exp\{\lambda(L(\alpha)-\ell(\alpha,w))\}
% \right]
% \le
% \log
% \mathbb E_{\alpha\sim\rho,w\sim \genDstro}
% \exp\{\lambda(L(\alpha)-\ell(\alpha,w))\}.
% $$
\end{remark}

\paragraph{Quadratic SLS specialization.}
If
$
L(\alpha)=\alpha_0+2g^\top\alpha+\alpha^\top G\alpha,
\, G\succeq0,
$
and $\rho=\mathcal N(\mu_\rho,\Sigma_\rho)$, then
$
\mathbb E_\rho e^{\lambda L(\alpha)}<\infty
$
whenever
$
I-2\lambda\Sigma_\rho^{1/2}G\Sigma_\rho^{1/2}\succ0.
$
This condition depends on the deployed posterior covariance, not on the prior.
However, it is conservative because it uses $L-\ell\le L$ and is mainly used as an example 
to highlight how to avoid radius collapse. Sharper certificates are given in the next section.

\section{Gaussian certificates}
\label{sec:explicit-certificates}
Corollary~\ref{thm:posterior-localized-pb-bound} is distribution-free at the PAC-Bayes level, 
but it is defined for values of $\lambda$ for which the posterior-localized exponential moment is finite.
To obtain explicit, computable certificates, we now specialize for 
%two disturbance models: 
the Gaussian disturbance, which yields an exact
closed-form expression for $\Lambda_\alpha(\lambda)$ by exploiting the
quadratic-Gaussian structure of the loss. A similar approach can be taken for bounded disturbances,
which yield a distribution-free certificate via Hoeffding's lemma~\citep{boucheron2013concentration}, which we do not cover here.
In both cases, the resulting bound is obtained by substituting the
corresponding certificate into Theorem~\ref{thm:pb_cgf}.

% \begin{remark}[Why Hanson-Wright does not suffice]
% The Hanson-Wright inequality~\citep{Rudelson2013} requires the
% coordinates of $w$ to be \emph{independent} sub-Gaussian,
% covering the special case $w \sim \mathcal{N}(0,\sigma^2 I)$
% but not our general trajectory-level covariance $\Sigma_w$.
% Even in the iid Gaussian case where Hanson-Wright applies,
% it yields a two-sided CGF bound valid only on a
% controller-dependent interval
% $\lambda \in [0,\,(C\sigma^2\|M(\alpha)\|_{\mathrm{op}}^2)^{-1})$,
% which collapses to zero under any posterior with unbounded support.
% Proposition~\ref{prop:exact-gaussian-integrand} avoids both obstacles.
% \end{remark}

\subsection{Gaussian disturbance certificates}
\label{sec:gaussian_cert}
Next, we derive an exact expression for the  inner exponential moment under a Gaussian
disturbance model. Unlike standard Hanson-Wright bounds~\citep{rudelson2013hanson,ziemann2025elementaryproofhansonwrightinequality}, 
which are typically stated for quadratic forms of vectors
with independent sub-Gaussian coordinates, our Gaussian certificate allows a full trajectory-level
covariance $\Sigma_w$.
 Thus, the disturbance sequence may be temporally correlated, and the certificate
depends on the induced closed-loop covariance
$
\Sigma_y(\alpha)=M(\alpha)\Sigma_wM(\alpha)^\top .
$

\begin{proposition}[Exact Gaussian integrand]
\label{prop:exact-gaussian-integrand}
Fix $\alpha\in\Theta$ and let $w\sim\mathcal D=\mathcal N(\mu_w,\Sigma_w)$, so that
$y_\alpha:=M(\alpha)w+m(\alpha)\sim\mathcal N(\mu_y(\alpha),\Sigma_y(\alpha))$ with
$\mu_y(\alpha)=M(\alpha)\mu_w+m(\alpha)$ and
$\Sigma_y(\alpha)=M(\alpha)\Sigma_w M(\alpha)^\top$, hence $
\ell(\alpha,w)=\|y_\alpha\|_2^2,  L(\alpha)=\operatorname{tr}(\Sigma_y(\alpha))+\|\mu_y(\alpha)\|_2^2
$. Define a shorthand $\overline{\Sigma}(\alpha) = I+2\lambda\Sigma_y(\alpha)$, then for every $\lambda\ge0$,
\begin{align}
  \mathbb E_{w\sim\mathcal D}\left [
  e^{\lambda(L(\alpha)-\ell(\alpha,w))}\right]
  = e^{\lambda L(\alpha)}
   %\det(I+2\lambda\Sigma_y(\alpha))^{-1/2}
    \det \left(\overline{\Sigma}(\alpha)\right)^{-1/2}
    %\exp\!\big(-\lambda\,\mu_y(\alpha)^\top(I+2\lambda\Sigma_y(\alpha))^{-1}\mu_y(\alpha)\big).
    e^{\big(-\lambda\,\mu_y(\alpha)^\top \overline{\Sigma}(\alpha)^{-1}\mu_y(\alpha)\big)}.
     \label{eq:exact-gaussian-integrand}
\end{align}
Equivalently, the centered CGF is
\begin{align}
  \Lambda_\alpha(\lambda)
  = \lambda L(\alpha)
    - \tfrac12\log\det(I+2\lambda\Sigma_y(\alpha))
    - \lambda\,\mu_y(\alpha)^\top(I+2\lambda\Sigma_y(\alpha))^{-1}\mu_y(\alpha).
    \label{eq:exact-gaussian-cgf}
\end{align}
\end{proposition}

\begin{proof}
Since 
$
y_\alpha=M(\alpha)w+m(\alpha)\sim \mathcal N(\mu_y(\alpha),\Sigma_y(\alpha)),
$
we can write
$$
\ell(\alpha,w)=\|y_\alpha\|_2^2,
\qquad
L(\alpha)=\operatorname{tr}(\Sigma_y(\alpha))+\|\mu_y(\alpha)\|_2^2.
$$
Hence
$$
\mathbb E_w e^{\lambda(L(\alpha)-\ell(\alpha,w))}
=
e^{\lambda L(\alpha)}
\mathbb E_{y_\alpha} e^{-\lambda \|y_\alpha\|_2^2}.
$$
For a Gaussian vector $y\sim\mathcal N(\mu,\Sigma)$ with $\Sigma \succ 0$ and any $\lambda\ge 0$, 
$$
\mathbb E e^{-\lambda \|y\|_2^2}
=
\det(I+2\lambda \Sigma)^{-1/2}
\exp\!\Big(
-\lambda\,\mu^\top (I+2\lambda\Sigma)^{-1}\mu
\Big)
$$
by \citep[Corollary 3.2a.2]{mathai1992quadratic}, but the same formula holds for 
$\Sigma \succeq 0$ due to $I + 2\lambda \Sigma$ being well defined,
which is a consequence of the one-sided CGF we consider (see Appendix \ref{app-sec:prop:exact-gaussian-integrand}).
Applying this identity with $(\mu,\Sigma)=(\mu_y(\alpha),\Sigma_y(\alpha))$ yields
\eqref{eq:exact-gaussian-integrand}. Taking logarithms gives \eqref{eq:exact-gaussian-cgf}.
\end{proof}

% For $y\sim\mathcal N(\mu,\Sigma)$, we use the standard Gaussian quadratic-form identity
% $$
% \mathbb E\exp(-y^\top A y)
% =
% \det(I+2\Sigma A)^{-1/2}
% \exp\!\left(
% -\mu^\top A(I+2\Sigma A)^{-1}\mu
% \right),
% \qquad A\succeq 0.
% $$
% Taking $A=\lambda I$ gives
% $$
% \mathbb E e^{-\lambda\|y\|_2^2}
% =
% \det(I+2\lambda\Sigma)^{-1/2}
% \exp\!\left(
% -\lambda\mu^\top(I+2\lambda\Sigma)^{-1}\mu
% \right).
% $$

Proposition~\ref{prop:exact-gaussian-integrand} gives an exact expression for the one-sided CGF of the squared affine-SLS loss under Gaussian disturbances. 
Its importance is twofold. First, it preserves the damping induced by the negative loss term $-\ell(\alpha,w)$. 
Second, the resulting expression is finite for every $\lambda\ge0$, since it depends on $I+2\lambda\Sigma_y(\alpha)$. Thus, in the Gaussian disturbance case, the one-sided Chernoff transform avoids a posterior-dependent finite radius. The certificate depends on the mean and covariance of the weighted closed-loop trajectory, $\mu_y(\alpha)$ and $\Sigma_y(\alpha)$, providing a direct closed-loop sensitivity interpretation.
% Note that in the above proof we exploit the fact that we are bounding a one sided loss, which enables us to avoid the 
% problem with $\lambda$ interval validity. This problem would emerge if we tried to use a standard quadratic 
% loss concentration like Hanson-Wright that controls deviations of quadratic forms around it's mean.

% \begin{remark}[Why Hanson-Wright does not apply]
% Hanson-Wright-type inequalities control two-sided deviations
% $|\ell(\alpha,w) - L(\alpha)|$ and yield a CGF bound valid only on
% a controller-dependent interval
% $\lambda \in [0, (C\sigma^2\|M(\alpha)\|_\mathrm{op}^2)^{-1})$,
% which collapses to zero for any posterior with unbounded support.
% The one-sided structure of our certificates --- bounding
% $L(\alpha) - \ell(\alpha,w)$ rather than its absolute value ---
% avoids this collapse: since $\ell(\alpha,w) \geq 0$,
% the one-sided exponential moment is bounded above by
% $e^{\lambda L(\alpha)}$, which is finite for every fixed $\alpha$
% and every $\lambda \geq 0$, with no controller-dependent radius.
% \end{remark}

\begin{corollary}[Quadratic upper bound for the Gaussian CGF]
\label{cor:gaussian-quadratic}
Under the assumptions of Proposition~\ref{prop:exact-gaussian-integrand}, for every $\alpha\in\Theta$
and every $\lambda\ge 0$,
\begin{align}
\Lambda_\alpha(\lambda)
\le
\lambda^2 \left( \|\Sigma_y(\alpha)\|_F^2
+
2 \mu_y(\alpha)^\top\Sigma_y(\alpha)\mu_y(\alpha) \right)\dfnrev \psi(\alpha,\lambda)
\label{eq:gaussian-quadratic-pointwise}
\end{align}
where
$$
\mu_y(\alpha):=M(\alpha)\mu_w+m(\alpha),
\qquad
\Sigma_y(\alpha):=M(\alpha)\Sigma_w M(\alpha)^\top.
$$
and therefore Theorem~\ref{thm:pb_cgf} yields the bound
% \begin{align}
% \mathbb E_{\rho}[L(\alpha)] \hspace{-0.15em}
% \le \hspace{-0.15em}
% \mathbb E_{\rho}[\widehat L(\alpha)]
% \hspace{-0.15em}
% +
% \hspace{-0.15em}
% \inf_{\lambda>0}
% \left\{
% \frac{\mathrm{KL}(\rho\Vert\pi)+\log(n/\delta)}{\lambda(n-1)}
% \hspace{-0.15em}
% +\hspace{-0.15em}\lambda
% \mathbb E_{\rho}\!\Big[
% \|\Sigma_y(\alpha)\|_F^2
% \hspace{-0.15em}
% +
% \hspace{-0.15em}
% 2\mu_y(\alpha)^\top\Sigma_y(\alpha)\mu_y(\alpha)
% \Big]
% \right\}.
% %\label{eq:gaussian-cor}
% \end{align}
\begin{align}
\mathbb E_{\rho}[L(\alpha)] \hspace{-0.15em}
\le
\mathbb E_{\rho}[\widehat L(\alpha)]
+
\inf_{\lambda>0}
\left\{
\frac{\mathrm{KL}(\rho\Vert\pi)+\log(n/\delta)}{\lambda(n-1)}
+
\lambda
\mathbb E_{\rho}\!\Big[
\phi(\alpha)
\Big]
\right\},
%\label{eq:gaussian-cor}
\end{align}
where $\phi(\alpha) \dfn \|\Sigma_y(\alpha)\|_F^2
+
2 \mu_y(\alpha)^\top\Sigma_y(\alpha)\mu_y(\alpha) $.
\end{corollary}

\begin{proof}
We sketch the proof here and provide the full one in Appendix~\ref{app-sec:proof-gaussian-quadratic}. We start from eq.~\eqref{eq:exact-gaussian-cgf}
% $$
% \Lambda_\alpha(\lambda)
% =
% \lambda \operatorname{tr}(\Sigma_y(\alpha))
% -\frac12 \log\det(I+2\lambda \Sigma_y(\alpha))
% +
% \lambda \|\mu_y(\alpha)\|_2^2
% -\lambda\,\mu_y(\alpha)^\top
% (I+2\lambda\Sigma_y(\alpha))^{-1}
% \mu_y(\alpha),
% $$
and group the terms with $\mu_y(\alpha)$ and $\Sigma_y(\alpha)$ separately.
For the mean term, we use the identity %\cite[Section 3.25, eq.~(165)]{cookbook} which for our case reads as
$
I-(I+2\lambda \Sigma_y(\alpha))^{-1}
=
2\lambda \Sigma_y(\alpha)(I+2\lambda\Sigma_y(\alpha))^{-1}
$
to obtain
\begin{align*}
% &\lambda \|\mu_y(\alpha)\|_2^2
% -\lambda\,\mu_y(\alpha)^\top
% (I+2\lambda\Sigma_y(\alpha))^{-1}
% \mu_y(\alpha) 
% =
2\lambda^2\,
\mu_y(\alpha)^\top
\Sigma_y(\alpha)(I+2\lambda\Sigma_y(\alpha))^{-1}
\mu_y(\alpha) \le
2\lambda^2
\mu_y(\alpha)^\top
\Sigma_y(\alpha)
\mu_y(\alpha).
\end{align*}
In the above, we used the fact that $(I+2\lambda\Sigma_y(\alpha))^{-1} \preceq I$.
For the covariance term, let $s_1,\dots,s_r$ denote the eigenvalues of $\Sigma_y(\alpha)$.
Then%, due to basic properties of the trace and the determinant of a PSD matrix, we have
\begin{align}
\lambda \operatorname{tr}(\Sigma_y(\alpha))
-\frac12 \log\det(I+2\lambda \Sigma_y(\alpha))
=
\sum_{i=1}^r \left(
\lambda s_i - \frac12 \log(1+2\lambda s_i)
\right).
\label{eq:gauss_PSD_trace_and_det}
\end{align}
A valid scalar inequality for $x \ge 0$ is $x-\frac12\log(1+2x)\le x^2$ (see Appendix~\ref{app-sec:proof-gaussian-quadratic}).
% \begin{align}
% \qquad \forall x\ge 0,
% \label{eq:ad-hoc-bound}
% \end{align}
% The validity of the bound can be checked by defining $f(x) = x^2 - x +\frac12\log(1+2x) $.
% We have $f(0) = 0$, hence the inequality is tight at zero. Furthermore, $f'(x)  =\frac{4x^2}{1+2x}$, 
% which is clearly nonnegative for any $x \ge 0$, proving our ad-hoc bound. It is also clear that the bound is tight 
% around zero, but it gets progressively worse for higher values.
Substituting $x=\lambda s_i$, we can bound~\eqref{eq:gauss_PSD_trace_and_det} by $\lambda^2 \|\Sigma_y(\alpha)\|_F^2,$
% $$
% \lambda \operatorname{tr}(\Sigma_y(\alpha))
% -\frac12 \log\det(I+2\lambda \Sigma_y(\alpha))
% \le
% \lambda^2 \sum_{i=1}^r s_i^2
% =
% \lambda^2 \|\Sigma_y(\alpha)\|_F^2,
% $$
which follows from the definition of the Frobenius norm.
Combining the two estimates proves \eqref{eq:gaussian-quadratic-pointwise}. Moreover, $\psi(\alpha,0) = \psi'(\alpha,0) = 0$ and $\psi(\alpha,\lambda)$ is clearly convex in $\lambda$ satisfying 
Assumption~\ref{ass:bounded_cgf} with $b = \infty$.
\end{proof}
Both terms of the certificate are closed-loop sensitivity quantities expressed through the induced covariance $\Sigma_y(\alpha)$.
The covariance term $\|\Sigma_y(\alpha)\|_F^2$  penalizes the overall disturbance amplification, while the mean term $\mu_y(\alpha)^\top \Sigma_y(\alpha)\mu_y(\alpha)$
penalizes the mean offset in the same induced metric.
\label{sec:Numerical}

\section{Deterministic mean-response deployment and data-driven bounds}
\label{sec:mean-response-deployment}
PAC-Bayes certifies randomized posterior risks, whereas control applications
typically deploy a deterministic controller. This stems from stringent testing and repeatability analysis.
In this section, we will therefore use the posterior $\rho$ as a certification object during learning but deploy only its
mean response $\mu_\rho$. First, note that the controller that achieves it is $K_{\rm mean} =\Phi_u(\mu_\rho)\Phi_x(\mu_\rho)^{-1}.$ %with 
%$\mu_\rho:=\mathbb E_{\alpha\sim\rho}[\alpha]$.
% $$
% \mu_\rho:=\mathbb E_{\alpha\sim\rho}[\alpha],
% \quad
% K_{\rm mean}=\Phi_u(\mu_\rho)\Phi_x(\mu_\rho)^{-1}.
% $$
Furthermore, it needs to be said that we will keep the \textit{posterior distribution} in the complexity term, 
blending stochastic certification and deterministic deployment, keeping the best of both worlds (so to speak) in a single formulation.
This also avoids applying PAC-Bayes directly to the Dirac posterior
$\delta_{\mu_\rho}$, which would generally have infinite KL divergence with
respect to a continuous prior. In the affine SLS parameterization, $M(\alpha)$ and $m(\alpha)$ 
are affine in $\alpha$. Hence, for every fixed disturbance trajectory $w$, the loss
$\ell(\alpha,w)=\|M(\alpha)w+m(\alpha)\|_2^2$
is convex quadratic in $\alpha$. 
Assuming
$\mathbb E_\genDstro\|w\|^2<\infty$, the population risk is
\begin{align}
L(\alpha)
=
\mathbb E_{w\sim \genDstro}[\ell(\alpha,w)]
=
\alpha^\top G_\alpha \alpha+2g_\alpha^\top\alpha+\alpha_0,
\qquad
G_\alpha \succeq0.
\end{align}
Similarly, the empirical risk has a quadratic form
\begin{align}
\widehat L_S(\alpha)
=
\alpha^\top \widehat G_S\alpha
+
2\widehat g_S^\top\alpha
+
\widehat\alpha_{0,S},\qquad
\widehat G_S \succeq 0.
\end{align}
We can now state the following proposition.
\begin{proposition}[Deterministic mean-response certificate]
\label{prop:mean_response_certificate}
Let $\rho$ be a posterior with a finite second moment, and denote the mean as $\mu_\rho$ and
the covariance as $\Sigma_\rho$. Suppose that, with probability at least
$1-\delta$, the randomized PAC-Bayes certificate
$
\mathbb E_{\alpha\sim\rho}[L(\alpha)]
\le
\mathbb E_{\alpha\sim\rho}[\widehat L_S(\alpha)]
+
\mathrm{Comp}(\rho,S,\delta)
$
holds. Then, on the same event,
\[
L(\mu_\rho)
\le
\widehat L_S(\mu_\rho)
+
\mathrm{Comp}(\rho,S,\delta)
+
\operatorname{tr}\!\left((\widehat G_S-G)\Sigma_\rho\right).
\]
\end{proposition}

\begin{proof}
Since $L$ is quadratic,
$
\mathbb E_{\alpha\sim\rho}[L(\alpha)]
=
L(\mu_\rho)+\operatorname{tr}(G\Sigma_\rho).
$
Likewise,
$
\mathbb E_{\alpha\sim\rho}[\widehat L_S(\alpha)]
=
\widehat L_S(\mu_\rho)
+
\operatorname{tr}(\widehat G_S\Sigma_\rho).
$
Substituting both identities into the randomized PAC-Bayes certificate and
rearranging gives the claim.
\end{proof}
Note that this bound is still an oracle bound, but the unknown term can be bounded via standard 
concentration inequalities~\citep{vershynin2018high,wainwright2019high}.
Proposition~\ref{prop:mean_response_certificate} shows that posterior learning
can be seen as a distributional regularization of deterministic empirical risk minimization (ERM). During training,
the empirical posterior risk decomposes as
$
\mathbb E_{\alpha\sim\rho}[\widehat L_S(\alpha)]
=
\widehat L_S(\mu_\rho)
+
\operatorname{tr}(\widehat G_S\Sigma_\rho)$,
so the posterior regularizes the deployed mean by encouraging a finite KL
neighborhood around it to have low empirical cost. After transferring the
certificate to the deterministic mean response, the full empirical spread
penalty is not paid as a deployment cost, it is only the empirical to population
curvature mismatch along $\Sigma_\rho$ that contributes. 
In particular, if the empirical and population curvatures agree along the
posterior covariance, i.e.
$
\operatorname{tr} (\widehat G_S-G)\Sigma_\rho =0,
$
then
$
L(\mu_\rho)
\le
\widehat L_S(\mu_\rho)
+
\mathrm{Comp}(\rho,S,\delta).
$

Thus, the deployed controller is deterministic, but its certificate is inherited
from the non-degenerate posterior surrounding it. This is the mechanism by
which PAC-Bayes provides a finite-sample certificate for mean deployment without
incurring the infinite KL cost of a Dirac posterior.

\begin{proposition}[Data-driven deterministic certificate]
\label{prop:data_driven_mean_certificate}
Consider the zero mean Gaussian disturbance setting
$
\bw_i \overset{\mathrm{i.i.d.}}{\sim} \mathcal N(0,\Sigma_w)$ with 
$
m(\alpha)=0,
$
and suppose that the covariance event
$
\mathcal E_\Sigma
=
%\left
\{
\|\widehat\Sigma_w-\Sigma_w\|_{\rm op}
\le
\epsilon_\Sigma
%\right
\}
$
holds. Define the data-driven sensitivity coefficient, analogous to Corollary~\ref{cor:gaussian-quadratic}, as
% \[
% \widetilde \Sigma_w
% :=
% \widehat\Sigma_w+\epsilon_\Sigma I
% \]
%and
$$
\widehat C_\rho
:=
\mathbb E_{\alpha\sim\rho}
\left[
\left\|
M(\alpha)(\widehat\Sigma_w+\epsilon_\Sigma I) M(\alpha)^\top
\right\|_F^2
\right].
$$
Let
% $
% A_\rho
% :=
% \sum_{i,j=1}^p
% (\Sigma_\rho)_{ij}
% (M_{i})^\top M_{j}.
% $
$
(M_G)_{ij} = \tr{(M_i^\top M_j)}
$
be the Gram matrix.
Then, on the intersection of the PAC-Bayes event and the covariance event, the deterministic mean response $\mu_\rho=\mathbb E_\rho[\alpha]$ satisfies
\begin{align}
L(\mu_\rho)
\le
\widehat L_S(\mu_\rho)
+
2\sqrt{
\frac{
\widehat C_\rho
\left(
\mathrm{KL}(\rho\|\pi)+\log(n/\delta_{\rm PB})
\right)
}{n-1}
}
+
\epsilon_\Sigma\operatorname{tr}(\Sigma_\rho M_G).
\label{eq:data-drive-bound-full}
\end{align}
Finally, if $\mathbb P(\mathcal E_\Sigma)\ge 1-\delta_\Sigma$, then the bound~\eqref{eq:data-drive-bound-full} holds with a probability of at least
$
1-\delta_{\rm PB}-\delta_\Sigma.
$
\end{proposition}
The above bound is no longer an oracle bound but a data-driven one.
A valid $\epsilon_\Sigma$ may be obtained from standard Gaussian sample
covariance concentration bounds, and we provide an explicit choice in Appendix~\ref{app:covariance_radius}.
For a common choice of a diagonal Gaussian posterior,
$
\Sigma_\rho=\operatorname{diag}(\sigma_1^2,\ldots,\sigma_p^2),
$
the mismatch term becomes
$
% \epsilon_\Sigma\operatorname{tr}(A_\rho)
% =
\epsilon_\Sigma
\sum_{k=1}^p
\sigma_k^2\|M_{k}\|_F^2.
$
Thus, the additional derandomization term is a weighted posterior-variance
penalty, with weights determined entirely by the SLS basis directions.
Moreover, even in the general case,  $M_{\mathrm{G}}$ can be computed offline.
\section{Numerical experiments}
\label{sec:experiments}

All experiments are conducted on a laptop with a 12th Gen Intel i7-12700H processor, 32\,GB of
RAM, running Ubuntu 24.04.4 LTS. Code is implemented in Julia~\citep{bezanson2017julia}
(version~1.12.6) using automatic differentiation via ForwardDiff~\citep{RevelsLubinPapamarkou2016}.

\subsection{Posterior learning versus point optimization in the low-data regime}
\label{app:double-integrator-low-data}

We consider the double-integrator system as in~\eqref{eq:LTI_basic} with
$$
A =
\begin{bmatrix}
1 & 1\\
0 & 1
\end{bmatrix},
\quad
B =
\begin{bmatrix}
0.5\\
1
\end{bmatrix},
\quad Q=\operatorname{diag}(1,0.1), \quad P=\operatorname{diag}(2,0.2), \quad R = 0.5,
$$
with horizon $T=10$ and Gaussian disturbance trajectory
$
w=(\bar x_0,\bar w_0,\ldots,\bar w_{T-1}) \sim \mathcal N(0,0.1I).
$
The finite-horizon quadratic cost is
$
\ell(x,u)
=
\sum_{t=0}^{T-1} \bar x_t^\top Q \bar x_t
+
\sum_{t=0}^{T-1} \bar u_t^\top R \bar u_t
+
\bar x_T^\top P \bar x_T.
$
% where
% $$
% Q=\operatorname{diag}(1,0.1),
% \qquad
% R=0.5,
% \qquad
% P=\operatorname{diag}(2,0.2).
% $$
The SLS parameterization has nullspace dimension $p=110$.
% Writing the weighted
% closed-loop response as
% $$
% M(\alpha)=M_0+\sum_{k=1}^p \alpha_k M_k,
% $$
The loss in SLS coordinates is 
$
\ell(\alpha,w)=\|M(\alpha)w\|_2^2$. 
Moreover, assume access to a dataset $S = \{w_i\}_{i =1}^n$.
\paragraph{Low-data mean-response deployment.}
Next, we test whether the deterministic mean response obtained from a PAC-Bayes
posterior acts as a useful finite-sample regularizer. We use the double-integrator
system described above. The SLS nullspace has dimension $p=110$. Thus, for small sample sizes $n\ll p$,
unregularized point optimization is expected to be statistically fragile. Hence, we treat the derandomized
PAC-Bayes version as a synthesis method for control algorithms.

We compare four deterministic deployments. The raw ERM baseline solves
$$
    \alpha_{\rm ERM}
    \in
    \arg\min_{\alpha}
    \widehat L_S(\alpha),
$$
while ridge (Tikhonov) ERM solves
$$
    \alpha_{\rm ridge}
    \in
    \arg\min_{\alpha}
    \widehat L_S(\alpha) + \gamma \|\alpha\|_2^2 .
$$
The proposed method optimizes a diagonal Gaussian posterior
$\rho=\mathcal N(\mu_\rho,\operatorname{diag}(\sigma_1^2,\cdots,\sigma_p^2))$ using the data-driven
deterministic mean-response certificate
$$
    \widehat L_S(\mu_\rho)
    +
    \operatorname{Comp}_{\rm data}(\rho)
    +
    \varepsilon_\Sigma\sum_{i=1}^p
\sigma_i^2\|M_{i}\|_F^2,
$$
and deploys the deterministic response $\mu_\rho$. As a prescient (clairvoyant) reference, we
also report an oracle derandomized PAC-Bayes baseline that uses the true covariance
$\Sigma_w$ and population curvature $G$ inside the certificate. This oracle is
not available in a data-driven setting, but it indicates how much conservatism is
introduced by covariance estimation. The prior is set to $\mathcal N(0,I).$ 
Moreover, for the epsilon term $ \varepsilon_\Sigma $, we need 
a covariance scale, which is assumed to be $\bar \sigma^2 = 0.1$.

All methods are evaluated on $N_{\rm test}=5000$ newly drawn disturbance trajectories.
We report the held-out test cost and the deterministic closed-loop sensitivity coefficient
$C(\alpha) := \|M(\alpha)\Sigma_w M(\alpha)^\top\|_F^2 = \|\Sigma_y(\alpha)\|_F^2$,
the squared Frobenius norm of the induced closed-loop covariance, evaluated at the deployed
response (the point solution for ERM/ridge, and the mean response $\mu_\rho$ for the PAC-Bayes
deployments). This is the quantity whose posterior average is the sensitivity
coefficient $C_\rho = \mathbb{E}_{\alpha\sim\rho}[C(\alpha)]$ used in the Gaussian certificate.

\begin{figure}[t]
    \centering
    \includegraphics[width=\linewidth]{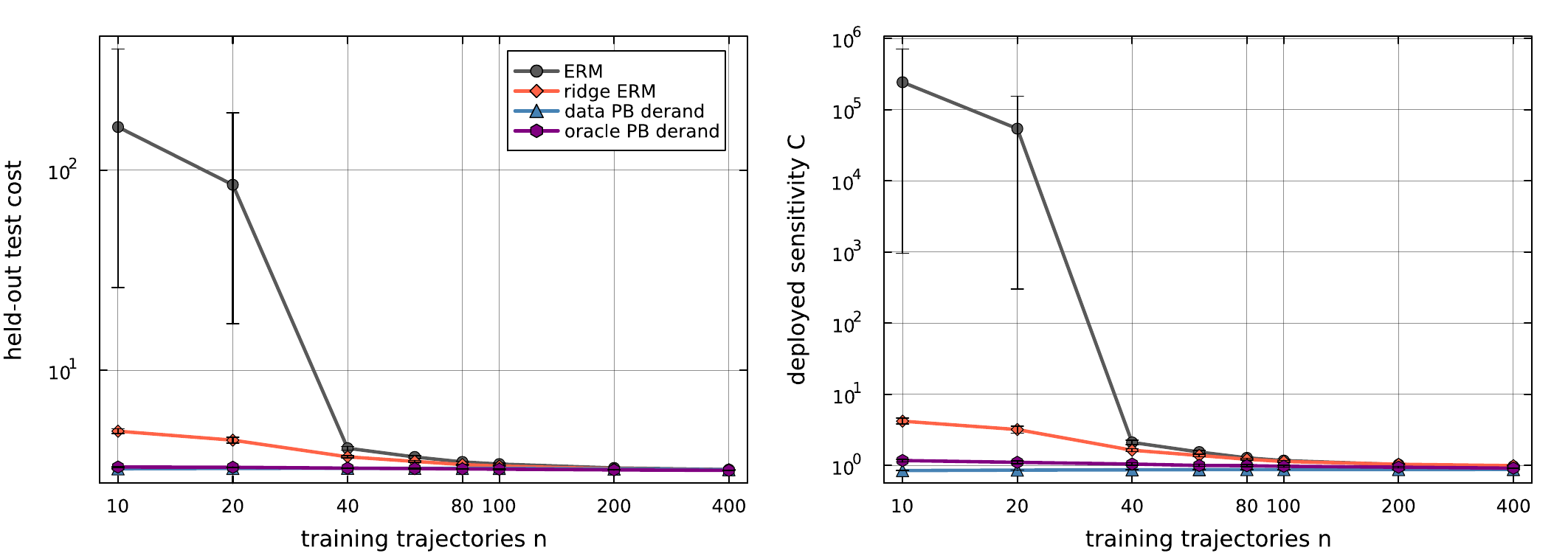}
    \caption{Double-integrator low-data sweep with $p=110$. Left: held-out test
    cost. Right: deterministic deployed sensitivity. Raw ERM is highly unstable
    when $n\ll p$. Ridge regularization stabilizes point optimization, but the
    data-driven derandomized PAC-Bayes mean response achieves substantially lower
    test cost and lower sensitivity in the low-data regime. As $n$ increases,
    all methods converge, indicating that the advantage is primarily finite-sample
    regularization.}
    \label{fig:di-data-driven-n-sweep}
\end{figure}

\begin{table}[t]
\centering
\caption{Double-integrator low-data sweep. Results are mean $\pm$ one standard
deviation over $20$ seeds. The PAC-Bayes columns report deterministic deployment
of the posterior mean response.}
\label{tab:di-data-driven-n-sweep}
\scriptsize
\setlength{\tabcolsep}{2.5pt}
\begin{tabular}{rcccc}
\toprule
& \multicolumn{4}{c}{Held-out test cost} \\
\cmidrule(lr){2-5}
$n$ & raw ERM & ridge ERM & data PB derand & oracle PB derand \\
\midrule
10  & $164.01 \pm 489.35$ & $4.96 \pm 0.33$ & $3.22 \pm 0.02$ & $3.28 \pm 0.03$ \\
20  & $84.37 \pm 216.30$  & $4.47 \pm 0.35$ & $3.23 \pm 0.02$ & $3.27 \pm 0.02$ \\
40  & $4.08 \pm 0.19$     & $3.69 \pm 0.11$ & $3.24 \pm 0.02$ & $3.24 \pm 0.03$ \\
100 & $3.40 \pm 0.09$     & $3.32 \pm 0.07$ & $3.23 \pm 0.03$ & $3.20 \pm 0.03$ \\
400 & $3.19 \pm 0.02$     & $3.18 \pm 0.02$ & $3.18 \pm 0.02$ & $3.16 \pm 0.02$ \\
\midrule
& \multicolumn{4}{c}{Deployed sensitivity $C(\alpha)$} \\
\cmidrule(lr){2-5}
$n$ & raw ERM & ridge ERM & data PB derand & oracle PB derand \\
\midrule
10  & $2.42{\times}10^5 \pm 1.03{\times}10^6$ & $4.19 \pm 1.01$ & $0.85 \pm 0.00$ & $1.17 \pm 0.13$ \\
20  & $5.39{\times}10^4 \pm 2.07{\times}10^5$ & $3.20 \pm 0.87$ & $0.86 \pm 0.01$ & $1.11 \pm 0.12$ \\
40  & $2.11 \pm 0.35$ & $1.64 \pm 0.18$ & $0.87 \pm 0.01$ & $1.05 \pm 0.10$ \\
100 & $1.17 \pm 0.11$ & $1.14 \pm 0.08$ & $0.88 \pm 0.01$ & $0.97 \pm 0.07$ \\
400 & $0.97 \pm 0.02$ & $0.99 \pm 0.02$ & $0.88 \pm 0.01$ & $0.92 \pm 0.01$ \\
\bottomrule
\end{tabular}
\end{table}

The advantage is most obvious when $n<p$. At $n=10$ and $n=20$, raw ERM
often finds responses that incur very large held-out variance. Ridge
regularization removes severe failures but remains substantially worse
than the data-driven PAC-Bayes mean response. The data-driven derandomized method
beats ridge ERM on every seed through $n=100$, and the gap closes as $n$ grows.
This is the expected behavior. With enough trajectories, point optimization becomes
well-conditioned, while in the low-data regime, the PAC-Bayes posterior acts as a
sensitivity-aware regularizer whose mean response gives a stable deterministic
controller.

The oracle derandomized baseline is useful as a reference but is not a deployable
method since it uses population quantities that are unavailable from data. The fact that
the data-driven method is close to this oracle in held-out cost, despite having a
more conservative certificate, suggests that the learned response is very good and that the data-driven 
PAC-Bayes bound can be thought of as a good synthesis method for low-data learning control algorithms.
The remaining gap is mainly a certification gap rather than a performance gap.

\paragraph{Numerical certificates}
Beyond synthesis quality, we evaluate the numerical tightness of the data-driven deterministic mean-response certificate. For each sample size, we compare
$$
B_{\mathrm{data}}(\rho)
=
\widehat L_S(\mu_\rho)
+
\operatorname{Comp}_{\mathrm{data}}(\rho)
+
\varepsilon_\Sigma
\sum_{k=1}^{p}\sigma_k^2\lVert M_k\rVert_F^2
$$
with the held-out estimate of $L(\mu_\rho)$. We also report a prescient evaluation of the \emph{same learned posterior} using the true disturbance covariance and population curvature. This diagnostic isolates the conservatism introduced by covariance estimation without changing the deployed mean response.

\begin{figure}[t]
    \centering
    \includegraphics[width=0.72\linewidth]{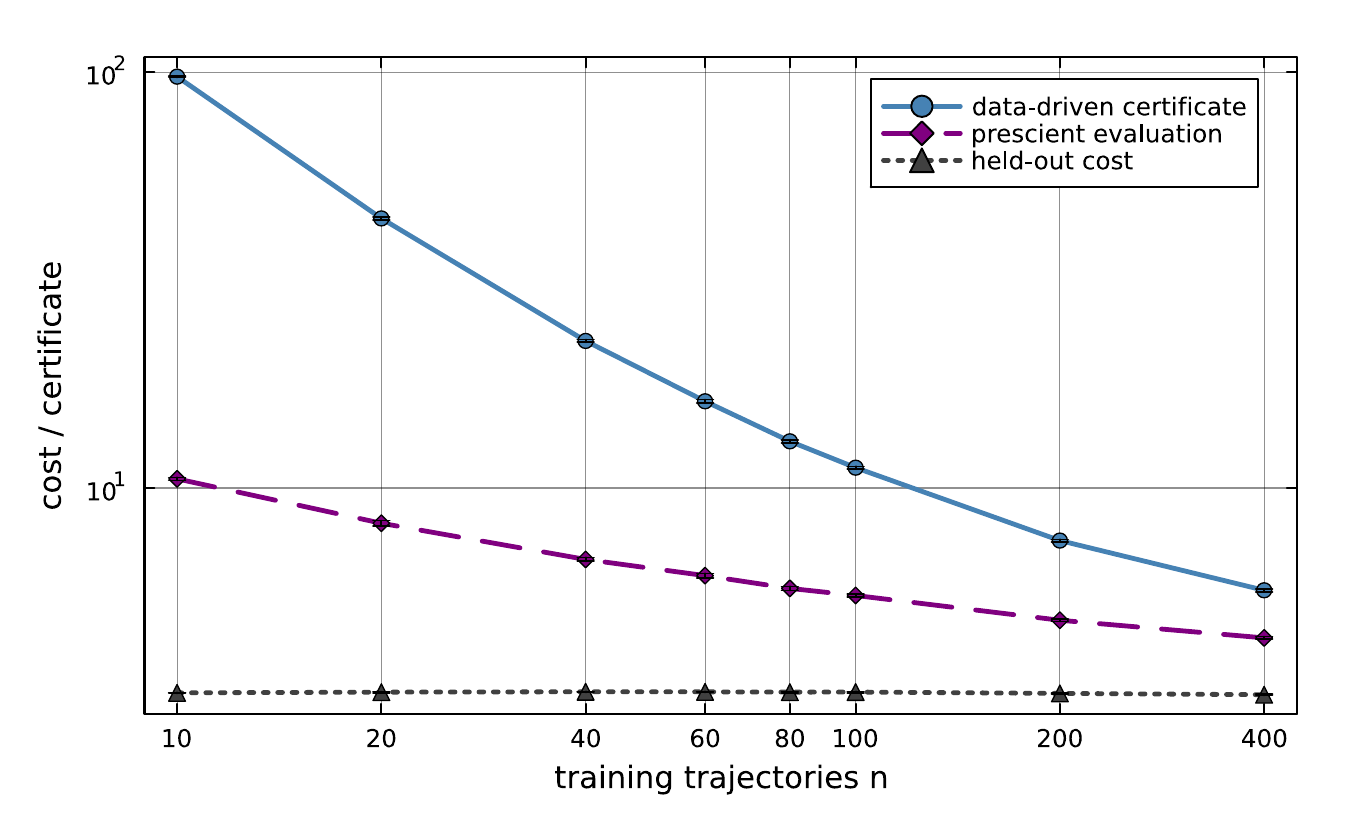}
    \caption{Numerical certificate tightness for the learned data-driven posterior. Curves show means over $20$ seeds. % and error bars are percentile 95\% bootstrap confidence intervals for the mean. 
    The prescient curve evaluates the same learned posterior with population quantities and is not used during training.}
    \label{fig:certificate-tightness}
\end{figure}

\begin{table}[t]
\centering
\caption{Numerical tightness of the data-driven deterministic mean-response certificate. All quantities are evaluated at the posterior learned by the data-driven derandomized objective. Results are mean $\pm$ one standard deviation over 20 seeds. The prescient column evaluates the same learned posterior using the true covariance and population curvature; it is diagnostic and is not available to the learning algorithm.}
\label{tab:certificate-tightness}
\small
\setlength{\tabcolsep}{4pt}
\begin{tabular}{rcccc}
\toprule
$n$ & held-out cost & data certificate & prescient certificate & certificate $-$ held-out \\
\midrule
$10$ & $3.22 \pm 0.02$ & $97.51 \pm 0.83$ & $10.52 \pm 0.20$ & $94.29 \pm 0.83$ \\
$20$ & $3.23 \pm 0.02$ & $44.48 \pm 0.83$ & $8.23 \pm 0.29$ & $41.25 \pm 0.82$ \\
$40$ & $3.24 \pm 0.02$ & $22.57 \pm 0.37$ & $6.73 \pm 0.16$ & $19.34 \pm 0.36$ \\
$100$ & $3.23 \pm 0.03$ & $11.19 \pm 0.19$ & $5.52 \pm 0.11$ & $7.96 \pm 0.19$ \\
$400$ & $3.18 \pm 0.02$ & $5.68 \pm 0.11$ & $4.36 \pm 0.07$ & $2.49 \pm 0.11$ \\
\bottomrule
\end{tabular}
\end{table}

The data-driven approach stays above the prescient one for all data points.
Its average certified gap contracts from $94.29$ at $n=10$ to $2.49$ at $n=400$. At $n=10$, the difference between the data-driven and prescient evaluations is $86.99$, showing that covariance-estimation conservatism dominates in the lowest-data regime. At $n=400$, this difference decreases to $1.32$, while the prescient certificate remains $1.18$ above the held-out cost. Thus, the certificate tightens steadily with data.

\section{Conclusion}
\label{sec:Conclusion}
We developed PAC-Bayesian certificates for finite-horizon quadratic
closed-loop control by placing posterior distributions directly over
feasible SLS response coordinates. This parameterization preserves
closed-loop achievability under unbounded  posteriors and exposes
the native loss as quadratic in both the disturbance trajectory and 
the free response coordinate. 

For Gaussian disturbances with arbitrary
covariance this yields an exact one-sided Chernoff transform and a tractable
sensitivity-based upper bound. The central message is that randomized certificates need
not imply randomized deployment. The posterior serves as a finite-KL certificate-carrying
object during learning, while the deployed controller is the deterministic mean response,
which pays only the empirical-to-population curvature mismatch along the posterior
covariance. Experiments on a double integrator confirm the resulting interpretation as a
sensitivity-aware finite-sample regularizer, with lower held-out cost and sensitivity than
ridge in the low-data regime, and the convergence of all methods as $n$ grows.

The sharpest certificates here rely on Gaussian or bounded disturbances and finite-horizon
linear systems without hard state-input constraints. Extending the approach to constrained
MPC  and model uncertainty is left for future work.

\bibliographystyle{plainnat}  
\bibliography{bib-pb-mpc}  
%%%%%%%%%%%%%%%%%%%%%%%%%%%%%%%%%%%%%%%%%%%%%%%%%%%%%%%%%%%%

\appendix

%$\section{Technical appendices and supplementary material}
%\include{Appendix/controller_dependent_chernoff_radius}
\section{Extended Proofs}
% \subsection{Proof of Proposition~\ref{prop:exact-posterior localized}}
% \label{app-sec:prop:poserior-localized}
% \input{Appendix/app_poof-posterior-localized}

\subsection{Proof of Proposition~\ref{prop:exact-gaussian-integrand}}
\label{app-sec:prop:exact-gaussian-integrand}
\begin{proof}
Since
$$
y_\alpha=M(\alpha)w+m(\alpha)\sim \mathcal N(\mu_y(\alpha),\Sigma_y(\alpha)),
$$
we can write
$$
\ell(\alpha,w)=\|y_\alpha\|_2^2,
\qquad
L(\alpha)=\operatorname{tr}(\Sigma_y(\alpha))+\|\mu_y(\alpha)\|_2^2.
$$
Hence
$$
\mathbb E_w \left[ e^{\lambda(L(\alpha)-\ell(\alpha,w))} \right]
=
e^{\lambda L(\alpha)}
\mathbb E_{y_\alpha} \left[ e^{-\lambda \|y_\alpha\|_2^2}\right].
$$
For a Gaussian vector $y\sim\mathcal N(\mu,\Sigma)$, with $\Sigma \succ 0$ and any $\lambda\ge 0$, 
$$
\mathbb E e^{-\lambda \|y\|_2^2}
=
\det(I+2\lambda \Sigma)^{-1/2}
\exp\!\Big(
-\lambda\,\mu^\top (I+2\lambda\Sigma)^{-1}\mu
\Big)
$$
by~\citep[Corollary 3.2a.2]{mathai1992quadratic}.
However, the same formula extends to case $\Sigma \succeq 0$ almost trivially because of the negative sign in the exponent. 
The key observation is that $I+2\lambda\Sigma$ is always well defined.
Indeed, write the spectral decomposition of $\Sigma\succeq 0$ as
$$
\Sigma
=
\begin{bmatrix}
Q_R & Q_0
\end{bmatrix}
\begin{bmatrix}
\Lambda_R & 0 \\
0 & 0
\end{bmatrix}
\begin{bmatrix}
Q_R^\top \\
Q_0^\top
\end{bmatrix},
$$
where $\Lambda_R\succ0$, $Q_R$ spans $\operatorname{Range}(\Sigma)$, and
$Q_0$ spans $\operatorname{Null}(\Sigma)$. Equivalently,
$$
\Sigma = Q_R\Lambda_RQ_R^\top,
$$
since the covariance has zero eigenvalues on $\operatorname{Null}(\Sigma)$.
% $$
% \Sigma = Q_R\Lambda_RQ_R^\top
% $$
% be the reduced spectral decomposition with $\Lambda_R\succ 0$, and let
% $Q_0$ span $\operatorname{Null}(\Sigma)$. 
Writing
$$
\mu = Q_RQ_R^\top\mu + Q_0Q_0^\top\mu = \mu_R + \mu_0,
$$
we have the representation
\begin{align}
y = \mu_0 + Q_R\left(Q_R^\top\mu + \Lambda_R^{1/2}z\right),
\qquad z\sim\mathcal N(0,I_r).
\end{align}
Thus
$$
\|y\|^2
=
\|\mu_0\|^2
+
\left\|Q_R^\top\mu+\Lambda_R^{1/2}z\right\|^2.
$$
Applying the nonsingular formula on the $r=\operatorname{rank}(\Sigma)$
dimensional range space gives
$$
\mathbb E e^{-\lambda\|y\|^2}
=
\det(I_r+2\lambda\Lambda_R)^{-1/2}
\exp\left(
-\lambda\|\mu_0\|^2
-\lambda\hat\mu_R^\top
(I_r+2\lambda\Lambda_R)^{-1}
\hat\mu_R
\right),
$$
where $\hat\mu_R=Q_R^\top\mu$ is the vector in reduced basis. Since the eigenvalues of
$I+2\lambda\Sigma$ on the nullspace are equal to one, this is equivalent
to the full-space expression
$$
\mathbb E e^{-\lambda\|y\|^2}
=
\det(I+2\lambda\Sigma)^{-1/2}
\exp\left(
-\lambda\mu^\top(I+2\lambda\Sigma)^{-1}\mu
\right).
$$
\end{proof}

\subsection{Proof of Corollary~\ref{cor:gaussian-quadratic}}
\label{app-sec:proof-gaussian-quadratic}
\begin{proof}
We start from eq.~\eqref{eq:exact-gaussian-cgf}
$$
\Lambda_\alpha(\lambda)
=
\lambda \operatorname{tr}(\Sigma_y(\alpha))
-\frac12 \log\det(I+2\lambda \Sigma_y(\alpha))
+
\lambda \|\mu_y(\alpha)\|_2^2
-\lambda\,\mu_y(\alpha)^\top
(I+2\lambda\Sigma_y(\alpha))^{-1}
\mu_y(\alpha),
$$
and group the terms with $\mu_y(\alpha)$ and $\Sigma_y(\alpha)$ separately.
For the mean term, we use \cite[Section 3.25, eq.~(165)]{cookbook} which for our case reads as
$$
I-(I+2\lambda \Sigma_y(\alpha))^{-1}
=
2\lambda \Sigma_y(\alpha)(I+2\lambda\Sigma_y(\alpha))^{-1},
$$
to obtain
\begin{align*}
&\lambda \|\mu_y(\alpha)\|_2^2
-\lambda\,\mu_y(\alpha)^\top
(I+2\lambda\Sigma_y(\alpha))^{-1}
\mu_y(\alpha) 
=
2\lambda^2\,
\mu_y(\alpha)^\top
\Sigma_y(\alpha)(I+2\lambda\Sigma_y(\alpha))^{-1}\mu_y(\alpha)
\end{align*}
Note that 
$$
(I+2\lambda\Sigma_y(\alpha))^{-1} \preceq I
$$
because $\Sigma_y(\alpha)$ is PSD and $\lambda \ge 0$.
Hence
\begin{align*}
2\lambda^2\,
\mu_y(\alpha)^\top
\Sigma_y(\alpha)(I+2\lambda\Sigma_y(\alpha))^{-1}\mu_y(\alpha) \le 2\lambda^2\,
\mu_y(\alpha)^\top
\Sigma_y(\alpha)\mu_y(\alpha) .
\end{align*}
For the covariance term, let $s_1,\dots,s_r$ denote the eigenvalues of $\Sigma_y(\alpha)$.
Then, due to basic properties of the trace and the determinant of a PSD matrix, we have
\begin{align}
\lambda \operatorname{tr}(\Sigma_y(\alpha))
-\frac12 \log\det(I+2\lambda \Sigma_y(\alpha))
=
\sum_{i=1}^r \left(
\lambda s_i - \frac12 \log(1+2\lambda s_i)
\right).
\label{eq:gauss_PSD_trace_and_det-app}
\end{align}

Next, we aim to find a simple bound for the above. It's easy to recognize the first part of the scalar inequality, 
which we bound with a quadratic term, i.e.
\begin{align}
x-\frac12\log(1+2x)\le x^2,
\qquad \forall x\ge 0,
\label{eq:ad-hoc-bound-app}
\end{align}
The validity of the bound can be checked by defining $f(x) = x^2 - x +\frac12\log(1+2x) $.
We have $f(0) = 0$, hence the inequality is tight at zero. Furthermore, $f'(x)  =\frac{4x^2}{1+2x}$, 
which is clearly nonnegative for any $x \ge 0$, proving our ad-hoc bound. It is also clear that the bound is tight 
around zero, but it gets progressively worse for higher values. However, optimization over $\lambda$ keeps it in the "good" regime.
Substituting $x=\lambda s_i$ into~\eqref{eq:ad-hoc-bound-app} and comparing to~\eqref{eq:gauss_PSD_trace_and_det-app}, we obtain
\begin{align}
\lambda \operatorname{tr}(\Sigma_y(\alpha))
-\frac12 \log\det(I+2\lambda \Sigma_y(\alpha))
\le
\lambda^2 \sum_{i=1}^r s_i^2
=
\lambda^2 \|\Sigma_y(\alpha)\|_F^2,
\label{eq:frobenius-bound-from-ad-hoc}
\end{align}
by the definition of the Frobenius norm.
Combining the two bounds proves \eqref{eq:gaussian-quadratic-pointwise}.
Define 
\begin{align}
\psi(\alpha,\lambda) \dfn 
%\Lambda_\alpha(\lambda)
%\le
\lambda^2 \|\Sigma_y(\alpha)\|_F^2
+
2\lambda^2 \mu_y(\alpha)^\top\Sigma_y(\alpha)\mu_y(\alpha)
\label{eq:gaussian-quadratic-pointwise}
\end{align}
Clearly $\psi(\alpha,0) = \psi'(\alpha,0) = 0$ as $\lambda^2$ is the leading term in the expression, and moreover, it is convex. Finally, it is
a valid object for Theorem~\ref{thm:pb_cgf} as stated in the paper which we inherited from~\citep{PAC-Bayes-Chernoff-Bounds-Unbounded-Losses}.
\end{proof}

\section{Detailed affine SLS formulation}
\label{app:sls}

This appendix expands the finite-horizon affine System Level Synthesis~\cite{wang2019_SLS} (SLS) formulation used in
the main text. The purpose is to clarify how the closed-loop response variables are defined 
and what achievability constraints they must satisfy. This is based on the known theory, and we mostly follow the exposition in~\citep{schuepp2025_SLS_affine}
We restrict this appendix to the classical pointwise 
view of SLS, as the main purpose is to get readers unfamiliar with this framework up to speed.
In addition, we provide vectorization of the constraints which are important for placing the posterior distribution over the response variables. 
In the subsequent Appendix we will show why the quadratic control loss becomes a
quadratic function of the disturbance trajectory. 

\subsection{Finite-horizon stacked dynamics}

Consider the finite-horizon linear system 
$$
\bar x_{t+1}=A \bar x_t+B \bar u_t+ \bar w_t,\quad t=0,\ldots,T-1,
$$
where $ \bar x_t\in\mathbb R^{n_x}$, $ \bar u_t\in\mathbb R^{n_u}$, and $\bar w_t\in\mathbb R^{n_x}$. We stack the
state, input, and disturbance trajectories as
\begin{align}
x :=
\begin{bmatrix}
\bar x_0\\
\bar x_1\\
\vdots\\
\bar x_T
\end{bmatrix}
\in\mathbb R^{(T+1)n_x},
\quad
u :=
\begin{bmatrix}
\bar u_0\\
\bar u_1\\
\vdots\\
\bar u_{T-1}
\end{bmatrix}
\in\mathbb R^{Tn_u},
\end{align}
and
\begin{align}
w :=
\begin{bmatrix}
\bar x_0\\
\bar w_0\\
\vdots\\
\bar w_{T-1}
\end{bmatrix}
\in\mathbb R^{(T+1)n_x}.
\label{eq:distrubance_traj_appendix}
\end{align}
Here $w$ contains both the initial condition and the process disturbances. Thus, one sample
$w_i$ in the PAC-Bayes formulation corresponds to one full finite-horizon disturbance
trajectory as defined in~\eqref{eq:distrubance_traj_appendix}.

\paragraph{Lifted system matrices.}
For the stacked trajectories
% $$
% x=\begin{bmatrix}x_0\\x_1\\ \vdots\\ x_T\end{bmatrix}\in\mathbb R^{(T+1)n_x},
% \quad
% u=\begin{bmatrix}u_0\\u_1\\ \vdots\\ u_{T-1}\end{bmatrix}\in\mathbb R^{Tn_u},
% $$
define
$$
\mathcal A:=I_{T+1}\otimes A,
\quad
\mathcal B:=
\begin{bmatrix}
I_T\otimes B\\
0_{n_x\times Tn_u}
\end{bmatrix}.
$$
Let $\mathcal Z\in\mathbb R^{(T+1)n_x\times (T+1)n_x}$ be the block down-shift operator. Then
$$
\mathcal Z \mathcal A x
=
\begin{bmatrix}
0\\
Ax_0\\
Ax_1\\
\vdots\\
Ax_{T-1}
\end{bmatrix},
\quad
\mathcal Z\mathcal B u
=
\begin{bmatrix}
0\\
Bu_0\\
Bu_1\\
\vdots\\
Bu_{T-1}
\end{bmatrix}.
$$
It is useful to define
$$
\mathcal Z_A := \mathcal Z\mathcal A,
\quad
\mathcal Z_B := \mathcal Z\mathcal B,
$$

Therefore, the stacked dynamics are
$$
x=\mathcal Z_A x+\mathcal Z_B u+w,
$$
or equivalently
$$
(I-\mathcal Z_A)x-\mathcal Z_B u=w.
$$
% Let $Z$ denote the block down-shift operator. For a stacked state trajectory $x$, the product
% $Z\mathcal A x$ is
% $$
% Z\mathcal A x
% =
% \begin{bmatrix}
% 0\\
% Ax_0\\
% Ax_1\\
% \vdots\\
% Ax_{T-1}
% \end{bmatrix},
% \quad
% Z\mathcal B u
% =
% \begin{bmatrix}
% 0\\
% Bu_0\\
% Bu_1\\
% \vdots\\
% Bu_{T-1}
% \end{bmatrix}.
% $$

% With this notation, the stacked dynamics are
% $$
% x = Z\mathcal A x + Z\mathcal B u + w,
% $$
% or equivalently,
% $$
% (I-Z\mathcal A)x - Z\mathcal B u = w.
% $$

\subsection{Affine closed-loop responses}
Rather than parameterizing the controller directly, SLS parameterizes the closed-loop maps from
disturbance trajectories to state and input trajectories. In the affine case, we write
$$
x=\clmapx w +\clmapxaff,
\quad
u=\clmapu w+\clmapuaff,
$$
or, equivalently,
$$
\begin{bmatrix}
x\\
u
\end{bmatrix}
=
\begin{bmatrix}
\clmapx\\
\clmapu
\end{bmatrix}w
+
\begin{bmatrix}
\clmapxaff\\
\clmapuaff
\end{bmatrix}.
$$
Here
$$
\clmapx\in\mathbb R^{(T+1)n_x\times (T+1)n_x},
\quad
\clmapu\in\mathbb R^{Tn_u\times (T+1)n_x},
$$
and
$$
\clmapxaff\in\mathbb R^{(T+1)n_x},
\quad
\clmapuaff\in\mathbb R^{Tn_u}.
$$

The matrices $\clmapx$ and $\clmapu$ are required to be block lower triangular, reflecting causality.
The state and input at time $t$ may depend only on the initial condition and disturbances up to the
appropriate time, not on future disturbances.
Substituting the affine response into the stacked dynamics gives
$$
(I-\mathcal Z_A)(\clmapx w + \clmapxaff)
-
\mathcal Z_B(\clmapu w +\clmapuaff)
=
w.
$$
Since this equality must hold for all disturbance trajectories $w$, the linear and affine parts must
match separately. Therefore, the SLS achievability constraints are
$$
(I-\mathcal Z_A)\clmapx-\mathcal Z_B\clmapu=I,
$$
and
$$
(I-\mathcal Z_A)\clmapxaff-\mathcal Z_B\clmapuaff=0.
$$
The first equation states that the linear response from disturbances to trajectories is compatible with
the dynamics. The second equation states that the affine offset is itself dynamically consistent.
More generally, if the stacked dynamics include a deterministic affine offset $d$,
$$
(I-\mathcal Z_A)x-\mathcal Z_B u=w+d,
$$
where $d$ denotes the stacked deterministic offset, then the affine achievability constraint becomes
$$
(I-\mathcal Z_A)\clmapxaff-\mathcal Z_B\clmapuaff=d.
$$
The purely linear dynamics considered above correspond to $d=0$.
Combining the linear and affine components, the affine SLS constraint can be written as
$$
\begin{bmatrix}
I-\mathcal Z_A & -\mathcal Z_B
\end{bmatrix}
\begin{bmatrix}
\clmapx & \clmapxaff\\
\clmapu & \clmapuaff
\end{bmatrix}
=
\begin{bmatrix}
I & d
\end{bmatrix}.
$$
For the nominal linear system without a deterministic offset, $d=0$.

\subsection{Controller recovery}
The SLS response variables describe closed-loop behavior. Under the usual finite-horizon causality
conditions, they also recover an implementable affine feedback controller.
Since $\clmapx$ is causal and has an invertible causal structure~\citep[Eq. 7a]{schuepp2025_SLS_affine}, we may solve
$$
w=\clmapx^{-1}(x-\clmapxaff).
$$
Substituting this into
$$
u=\clmapu w+\clmapuaff
$$
gives
$$
u
=
\clmapu\clmapx^{-1}(x-\clmapxaff)+\clmapuaff.
$$
Therefore, the corresponding affine controller is
$$
u=K x+k,
$$
where
$$
K:=\clmapu\clmapx^{-1},
\quad
k:=\clmapuaff-\clmapu\clmapx^{-1}\clmapxaff.
$$
Thus, any achievable affine SLS response induces an implementable causal affine controller.
Conversely, any causal affine controller applied to the finite-horizon linear system induces closed-loop
responses $(\clmapx,\clmapu,\clmapxaff,\clmapuaff)$ satisfying the SLS achievability constraints. In this sense,
the SLS equations parameterize exactly the achievable finite-horizon closed-loop responses.

\subsection{Weighted trajectory and quadratic loss}
Let the finite-horizon quadratic control cost be defined by positive semidefinite state weights and
positive definite input weights. Define the stacked weighting matrices
$$
\bar Q :=
\operatorname{blkdiag}(Q_0,Q_1,\ldots,Q_T),
\quad
\bar R :=
\operatorname{blkdiag}(R_0,R_1,\ldots,R_{T-1}).
$$
For example, one may take $Q_t=Q$ for $t<T$, $Q_T=P$, and $R_t=R$.
Define the weighted closed-loop trajectory
$$
y(x,u)
:=
\begin{bmatrix}
\bar Q^{1/2}x\\
\bar R^{1/2}u
\end{bmatrix}.
$$
Substituting the affine SLS response gives
$$
y(\theta,w)
=
\begin{bmatrix}
\bar Q^{1/2}(\clmapx w+\clmapxaff)\\
\bar R^{1/2}(\clmapu w+\clmapuaff)
\end{bmatrix}.
$$
Therefore
$$
y(\theta,w)=M(\theta)w+m(\theta),
$$
where
$$
M(\theta)
:=
\begin{bmatrix}
\bar Q^{1/2}\clmapx\\
\bar R^{1/2}\clmapu
\end{bmatrix},
\quad
m(\theta)
:=
\begin{bmatrix}
\bar Q^{1/2}\clmapxaff\\
\bar R^{1/2}\clmapuaff
\end{bmatrix}.
$$
The finite-horizon quadratic trajectory cost is then
$$
\ell(\theta,w)
:=
\|y(\theta,w)\|_2^2
=
\|M(\theta)w+m(\theta)\|_2^2.
$$
Expanding,
$$
\ell(\theta,w)
=
w^\top M(\theta)^\top M(\theta)w
+
2m(\theta)^\top M(\theta)w
+
m(\theta)^\top m(\theta).
$$
Equivalently, defining
$$
M_c(\theta):=M(\theta)^\top M(\theta),
\quad
c(\theta):=M(\theta)^\top m(\theta),
$$
we have
$$
\ell(\theta,w)
=
w^\top M_c(\theta)w
+
2c(\theta)^\top w
+
m(\theta)^\top m(\theta).
$$

This is the key structural property used in the main text: under affine SLS, the closed-loop quadratic
control cost is a quadratic function of the disturbance trajectory $w$.

\subsection{Vectorized SLS constraints}
\label{app:sls_vectorized_constraints}
We vectorize the SLS response variables to expose the affine feasible set as a standard finite-dimensional linear system $H\vartheta=h$. This allows us to parameterize all feasible responses by a particular solution plus a nullspace basis, $\vartheta=\vartheta_0+N\alpha$. Moreover, it allows us to place PAC-Bayesian priors and posteriors directly on the free coordinates $\alpha$.
\label{app:sls-vectorization}

We now describe how the affine SLS achievability constraints can be written as a finite-dimensional
linear system. Recall the affine closed-loop response
$$
x=\clmapx w +\clmapxaff,\quad u=\clmapu w+\clmapuaff,
$$
with achievability constraints
$$
(I-\mathcal Z_A)\clmapx-\mathcal Z_B\clmapu=I,
$$
and
$$
(I-\mathcal Z_A)\clmapxaff-\mathcal Z_B\clmapuaff=d,
$$
where $d=0$ for the nominal linear system without an additional deterministic offset.

Let
$
\mathcal F :=
\begin{bmatrix}
I-\mathcal Z_A & -\mathcal Z_B
\end{bmatrix}.
$
Then the two achievability equations become
$$
\mathcal F
\begin{bmatrix}
\clmapx\\
\clmapu
\end{bmatrix}
=
I,
\quad
\mathcal F
\begin{bmatrix}
\clmapxaff\\
\clmapuaff
\end{bmatrix}
=
d.
$$
Define
$$
\Phi :=
\begin{bmatrix}
\clmapx\\
\clmapu
\end{bmatrix},
\quad
\phi :=
\begin{bmatrix}
\clmapxaff\\
\clmapuaff
\end{bmatrix}
$$
to obtain
$
\mathcal F\Phi=I,
$
and
$
\mathcal F\phi=d.
$
Using the identity
$$
\operatorname{vec}(ABC)
=
(C^\top\otimes A)\operatorname{vec}(B),
$$
the constraint $\mathcal F\Phi=I$ can be vectorized as (see~\citep{cookbook})
$$
% (I^\top\otimes \mathcal F)\operatorname{vec}(\Phi)
% =
% \operatorname{vec}(I) \rightarrow 
(I\otimes \mathcal F)\operatorname{vec}(\Phi)
=
\operatorname{vec}(I).
$$
The affine-offset constraint is already linear
$
\mathcal F\phi=d.
$
Thus, if we define the full vector of response variables
$$
\vartheta
:=
\begin{bmatrix}
\operatorname{vec}(\Phi)\\
\phi
\end{bmatrix},
$$
the achievability constraints can be written compactly as
$
H \vartheta=  h,
$
where
$$
 H :=
\begin{bmatrix}
I\otimes \mathcal F & 0\\
0 & \mathcal F
\end{bmatrix},
\quad
 h :=
\begin{bmatrix}
\operatorname{vec}(I)\\
d
\end{bmatrix}.
$$

The response matrices $\clmapx$ and $\clmapu$ are block lower triangular due to causality.
Rather than vectorizing all entries and imposing additional zero constraints, one may vectorize only the
causal entries. Let $S$ be a selection matrix that maps the reduced vector of causal entries
$\vartheta_\Phi$ to the full vectorization:
$$
\operatorname{vec}(\Phi)=S\vartheta_\Phi .
$$
Then the decision variable becomes
$$
\vartheta :=
\begin{bmatrix}
\vartheta_\Phi\\
\phi
\end{bmatrix},
$$
and the reduced SLS constraints are
$$
H\vartheta=h,
$$
with
$$
H :=
\begin{bmatrix}
(I\otimes \mathcal F)S & 0\\
0 & \mathcal F
\end{bmatrix},
\quad
h :=
\begin{bmatrix}
\operatorname{vec}(I)\\
d
\end{bmatrix}.
$$
This is the finite-dimensional linear system used for the nullspace parameterization. If
$\vartheta_0$ is any solution of $H\vartheta=h$ and $N$ is a basis for $\ker(H)$, then every feasible
causal affine SLS response is represented as
\begin{align}
\vartheta(\alpha)=\vartheta_0+N\alpha.
\label{app:sls_nullspace_eq}
\end{align}
\paragraph{Causality and selection matrices.}
The selection matrix $S$ is only a helper matrix to remove entries of $\clmapx$ and
$\clmapu$ that are structurally zero due to causality. Equivalently, one may vectorize all entries and add
linear equality constraints enforcing the noncausal blocks to be zero. Both approaches lead to a
linear system of the form $H\vartheta=h$, and therefore both admit the same nullspace
parameterization.

\section{Mean-response deployment and quadratic loss structure}
\label{app:sls_mean_reponse}

The goal of this section is to show that the affine nature of $\vartheta$ implies affine dependence in $M(\alpha)$, which 
induces quadratic $L(\alpha)$. Hence, the quadratic structure of $L(\alpha)$ is a \emph{consequence}, not an assumption.
Then we show how this ties into controller derandomization via mean deployment.

We start in a somewhat backward order by first showing a more general statement that a quadratic $L(\alpha)$ 
admits a natural randomization. %Then we show that this is given in our formulation.

\begin{proposition}
\label{prop:quadratic_decomposition}
Let 
$L(\alpha) = \alpha^\top G \alpha + 2g^\top \alpha + \alpha_0$ with $G = G^\top$. 
Assume that $\rho$ has a finite second moment and denote
$\E_\rho[\alpha] = \mu_\rho$, $\operatorname{Cov} (\alpha) = \Sigma_\rho $. 
Then 
$$
\E_{\alpha \sim \rho}[L(\alpha)] = L(\mu_\rho) + \operatorname{tr}(G \Sigma_\rho).
$$
\end{proposition}
\begin{proof}
Rewrite $L(\alpha) = \operatorname{tr}(\alpha^\top G \alpha) + 2g^\top \alpha + \alpha_0$. Due to the cyclicity of the trace, we have $\operatorname{tr}(\alpha^\top G \alpha) =
\operatorname{tr}(G \alpha\alpha^\top )$. The linearity of both expectation and trace gives 
$\E_\rho[\operatorname{tr}(G \alpha\alpha^\top )] = \operatorname{tr}(G E_\rho[ \alpha \alpha^\top ])$. Next, we characterize $E_\rho[ \alpha \alpha^\top ]$.
By definition, $\operatorname{Cov}(\alpha) \dfn E_\rho[ (\alpha - \mu_\rho)(\alpha - \mu_\rho)^\top] =  E_\rho[ \alpha \alpha^\top ] - \mu_\rho \mu_\rho^\top$ after simple algebraic manipulations. Moreover, it follows that $ E_\rho[ \alpha \alpha^\top ] = \Sigma_\rho + \mu_\rho \mu_\rho^\top$. We have all the ingredients to compute 
$\E_\rho[L(\alpha)] = \operatorname{tr}(G\E_\rho[\alpha\alpha^\top]) +  2g^\top \E_\rho [\alpha] + \alpha_0 = \operatorname{tr}(G( \operatorname{Cov}(\alpha) + \mu_\rho \mu_\rho^\top)) +  2g^\top \mu_\rho + \alpha_0  = \operatorname{tr}(G\Sigma_\rho) +  \operatorname{tr}(G\mu_\rho \mu_\rho^\top)  +  2g^\top \mu_\rho + \alpha_0 $.
Again, using the cyclicity of the trace, we have $\operatorname{tr}(G \mu_\rho \mu_\rho^\top) = \operatorname{tr}( \mu_\rho^\top G\mu_\rho)  = \mu_\rho^\top G \mu_\rho$. The last equality is because  $\mu_\rho^\top G\mu_\rho$ is a scalar. Finally, recognizing the original definition of $L(\alpha)$ and plugging in $\mu_\rho$, we obtain the final form.
\end{proof}

Our next goal is to show that the SLS parametrization naturally results in a convex quadratic function in $\alpha$, and the above proposition is readily applicable.
% Expand~\eqref{app:sls_nullspace_eq}  to 
% \begin{align}
% \vartheta(\alpha)=\vartheta_0+\sum_{k=1}^p\alpha_k N_k.
% \end{align}
% where $N_k \in  \Re^{d_\vartheta}$ are the columns of $N \in \Re^{d_\vartheta \times p}$, and $\alpha_k$ are elements of the vector $\alpha \in \Re^p$.
% This implies the affine structure of $M(\alpha)$ and $m(\alpha)$. Note that we had $M(\theta)$ before, but now $\theta$ is 
% constrained to be a function of $\alpha$. Instead of cluttering the notation, we simply write $M(\alpha)$ directly. The same 
% goes for $\clmapx(\alpha)$ and other relevant variables.

\subsection{From nullspace coordinates to affine cost maps}
\label{app:nullspace_to_cost_maps}

We make explicit how the nullspace parameterization of the SLS constraints induces affine
dependence of the cost maps in the free coordinate $\alpha$.

Let
$$
d_x := (T+1)n_x,\quad d_u := Tn_u,\quad d_w := (T+1)n_x .
$$
Then
$$
x\in\mathbb R^{d_x},\quad u\in\mathbb R^{d_u},\quad w\in\mathbb R^{d_w}.
$$
The SLS response variables have dimensions
$$
\clmapx\in\mathbb R^{d_x\times d_w},\quad
\clmapu\in\mathbb R^{d_u\times d_w},
$$
and
$$
\clmapxaff\in\mathbb R^{d_x},\quad
\clmapuaff\in\mathbb R^{d_u}.
$$
We collect the response variables into the vector
$$
\vartheta
:=
\begin{bmatrix}
\operatorname{vec}(\clmapx)\\
\operatorname{vec}(\clmapu)\\
\clmapxaff\\
\clmapuaff
\end{bmatrix}
\in \R^{d_\vartheta}
$$
where $d_\vartheta = d_x d_w + d_u d_w + d_x + d_u$.
The vectorized SLS constraints can be written as $H\vartheta=h$.
Let $\vartheta_{0}$ be one feasible solution, and let
$$
N=
\begin{bmatrix}
N_1 & \cdots & N_p
\end{bmatrix}
 \in \Re^{d_\vartheta \times p}
$$
be a basis for $\ker(H)$. Then every feasible response can be written as
$$
\vartheta(\alpha)=\vartheta_{0}+N\alpha,
\quad
\alpha\in\mathbb R^p.
$$
Equivalently,
\begin{align}
\vartheta(\alpha)=\vartheta_0+\sum_{k=1}^p\alpha_k N_k.
\end{align}
We partition the particular solution as
\begin{align}
\vartheta_{0}
=
\begin{bmatrix}
\operatorname{vec}(\clmapx^{[0]})\\
\operatorname{vec}(\clmapu^{[0]})\\
\clmapxaff^{[0]}\\
\clmapuaff^{[0]}
\end{bmatrix}.
\end{align}
Similarly, each nullspace direction is partitioned as
\begin{align}
N_k
=
\begin{bmatrix}
\operatorname{vec}(\clmapx^{[k]})\\
\operatorname{vec}(\clmapu^{[k]})\\
\clmapxaff^{[k]}\\
\clmapuaff^{[k]}
\end{bmatrix},
\quad k=1,\ldots,p.
\end{align}
Here $\clmapx^{[k]}$ and $\clmapu^{[k]}$ are the full response matrices with their respective dimensions
$$
\clmapx^{[k]}\in\mathbb R^{d_x\times d_w},
\quad
\clmapu^{[k]}\in\mathbb R^{d_u\times d_w}.
$$
Taking the corresponding blocks of $\vartheta(\alpha)$ gives
\begin{equation}
\begin{split}
\operatorname{vec}(\clmapx(\alpha))
&=
\operatorname{vec}(\clmapx^{[0]})
+
\sum_{k=1}^p \alpha_k \operatorname{vec}(\clmapx^{[k]}),
\\
\operatorname{vec}(\clmapu(\alpha))
&=
\operatorname{vec}(\clmapu^{[0]})
+
\sum_{k=1}^p \alpha_k \operatorname{vec}(\clmapu^{[k]}),
\\
\clmapxaff(\alpha)
&=
\clmapxaff^{[0]}
+
\sum_{k=1}^p \alpha_k \clmapxaff^{[k]},
\\
\clmapuaff(\alpha)
&=
\clmapuaff^{[0]}
+
\sum_{k=1}^p \alpha_k \clmapuaff^{[k]}.
\end{split}
\end{equation}

Applying $\text{vec}^{-1}$ to first two identities yields
\begin{align}
\clmapx(\alpha)
&=
\clmapx^{[0]}
+
\sum_{k=1}^p \alpha_k \clmapx^{[k]},
\\
\clmapu(\alpha)
&=
\clmapu^{[0]}
+
\sum_{k=1}^p \alpha_k \clmapu^{[k]}.
\end{align}
Therefore, each SLS response component is affine in $\alpha$.
Now define the weighted cost maps
$$
M(\alpha)
:=
\begin{bmatrix}
\bar Q^{1/2}\clmapx(\alpha)\\
\bar R^{1/2}\clmapu(\alpha)
\end{bmatrix},
\quad
m(\alpha)
:=
\begin{bmatrix}
\bar Q^{1/2}\clmapxaff(\alpha)\\
\bar R^{1/2}\clmapuaff(\alpha)
\end{bmatrix}.
$$
Let
$
d_y:=d_x+d_u.
$
Then
$$
M(\alpha)\in\mathbb R^{d_y\times d_w},
\quad
m(\alpha)\in\mathbb R^{d_y}.
$$
Substituting the affine response expansion gives
\begin{align}
M(\alpha)
&=
\begin{bmatrix}
\bar Q^{1/2}\clmapx^{[0]}\\
\bar R^{1/2}\clmapu^{[0]}
\end{bmatrix}
+
\sum_{k=1}^p
\alpha_k
\begin{bmatrix}
\bar Q^{1/2}\clmapx^{[k]}\\
\bar R^{1/2}\clmapu^{[k]}
\end{bmatrix},
\\
m(\alpha)
&=
\begin{bmatrix}
\bar Q^{1/2}\clmapxaff^{[0]}\\
\bar R^{1/2}\clmapuaff^{[0]}
\end{bmatrix}
+
\sum_{k=1}^p
\alpha_k
\begin{bmatrix}
\bar Q^{1/2}\clmapxaff^{[k]}\\
\bar R^{1/2}\clmapuaff^{[k]}
\end{bmatrix}.
\end{align}
Define
\begin{align}
M_{0}
:=
\begin{bmatrix}
\bar Q^{1/2}\clmapx^{[0]}\\
\bar R^{1/2}\clmapu^{[0]}
\end{bmatrix},
\quad
M_{k}
:=
\begin{bmatrix}
\bar Q^{1/2}\clmapx^{[k]}\\
\bar R^{1/2}\clmapu^{[k]}
\end{bmatrix},
\end{align}
and
\begin{align}
m_{0}
:=
\begin{bmatrix}
\bar Q^{1/2}\clmapxaff^{[0]}\\
\bar R^{1/2}\clmapuaff^{[0]}
\end{bmatrix},
\quad
m_{k}
:=
\begin{bmatrix}
\bar Q^{1/2}\clmapxaff^{[k]}\\
\bar R^{1/2}\clmapuaff^{[k]}
\end{bmatrix},\end{align}
to obtain
\begin{align}
M(\alpha)
&=
M_{0}
+
\sum_{k=1}^p \alpha_k M_{k},
\\
m(\alpha)
&=
m_{0}
+
\sum_{k=1}^p \alpha_k m_{k}.
\label{eq:affine_M_in_alpha}
\end{align}
Thus $M(\alpha)$ and $m(\alpha)$ are affine functions of the SLS coordinate
$\alpha$.

\subsection{Quadratic loss in the SLS coordinate}
\label{app:quadraticity_alpha}
We now show that the finite-horizon quadratic control loss is a convex quadratic
function of $\alpha$ for each fixed disturbance trajectory $w$.
For fixed $w\in\mathbb R^{d_w}$, define
$$
a_0(w):=M_{0}w+m_{0}\in\mathbb R^{d_y},
$$
and
$$
\mathcal A(w)
:=
\begin{bmatrix}
M_{1}w+m_{1} &
M_{2}w+m_{2} &
\cdots &
M_{p}w+m_{p}
\end{bmatrix}
\in\mathbb R^{d_y\times p}.
$$
Using the affine expansions of $M(\alpha)$ and $m(\alpha)$, we obtain
\begin{align}
M(\alpha)w+m(\alpha)
&=
\left(
M_{0}
+
\sum_{k=1}^p \alpha_k M_{k}
\right)w
+
m_{0}
+
\sum_{k=1}^p \alpha_k m_{k}
\\
&=
\underbrace{M_{0}w+m_{0}}_{a_0(w)}
+
\underbrace{\sum_{k=1}^p
\alpha_k\left(M_{k}w+m_{k}\right)}_{\mathcal A(w) \alpha}
\\
&=
a_0(w)+\mathcal A(w)\alpha.
\end{align}
For a fixed $w$, we have
\begin{align*}
\ell(\alpha,w)
&=
\|M(\alpha)w+m(\alpha)\|_2^2
=
\|\mathcal A(w)\alpha+a_0(w)\|_2^2
\\
&=
\alpha^\top \mathcal A(w)^\top \mathcal A(w)\alpha
+
2a_0(w)^\top\mathcal A(w)\alpha
+
a_0(w)^\top a_0(w).
\end{align*}
Hence, for every fixed $w$, the map $\alpha\mapsto \ell(\alpha,w)$ is a quadratic function. Here we also note a biquadratic structure. The loss $\ell(w,\alpha)$ is quadratic in $\alpha$
for fixed $w$ and quadratic in $w$ for a fixed $\alpha$.
Finally, assuming the relevant expectations are finite, the population risk is
\begin{align}
L(\alpha)
:=
\mathbb E_{w\sim \genDstro}[\ell(\alpha,w)]
=
\alpha^\top G_\alpha\alpha
+
2g_\alpha^\top\alpha
+
c_\alpha,
\end{align}
where
$
G_\alpha
:=
\mathbb E_{w\sim \genDstro}
\left[
\mathcal A(w)^\top\mathcal A(w)
\right]
\in\mathbb R^{p\times p},
$
$
g_\alpha
:=
\mathbb E_{w\sim \genDstro}
\left[
\mathcal A(w)^\top a_0(w)
\right]
\in\mathbb R^p,
$
and
$
c_\alpha
:=
\mathbb E_{w\sim \genDstro}
\left[
a_0(w)^\top a_0(w)
\right]
\in\mathbb R.
$
Moreover, $G_\alpha\succeq0$
% Indeed, for any $v\in\mathbb R^p$,
% \begin{align}
% v^\top G_\alpha v
% &=
% \mathbb E_{w\sim \genDstro}
% \left[
% v^\top \mathcal A(w)^\top\mathcal A(w)v
% \right]
% \\
% &=
% \mathbb E_{w\sim \genDstro}
% \left[
% \|\mathcal A(w)v\|_2^2
% \right]
% \\
% &\ge 0.
% \end{align}
,thus $L(\alpha)$ is a convex quadratic function in the SLS coordinate $\alpha$.

\subsection{Deterministic deployment (derandomization)}
\label{app:deterministic-deployment}
The PAC-Bayes certificate is naturally stated for the randomized posterior risk
$
\mathbb E_{\alpha\sim\rho}[L(\alpha)].
$
In many control applications, it is highly preferable to deploy a single deterministic controller. This usually stems
from certification and testing requirements. In the SLS quadratic setting, the natural choice is to deploy the \emph{posterior mean response}. 
This subsection shows that the mean response is feasible, inherits the PAC-Bayes certificate, and induces an
implementable controller. This is a consequence of the convexity of the cost in response space, but we also  
show how it can be improved beyond the simple application of Jensen's inequality due to  quadratic structure. 
Jensen's gap is exactly $\tr(G_\alpha\Sigma_\rho).$

Let $\rho$ be a posterior distribution over $\alpha$. Assume that $\rho$ has a finite second moment 
and define its mean and covariance by
$$
\mu_\rho:=\mathbb E_{\alpha\sim\rho}[\alpha],
\quad
\Sigma_\rho:=\operatorname{Cov}_{\alpha\sim\rho}(\alpha).
$$
Because $\vartheta(\alpha)$ is affine in $\alpha$, the posterior mean response is exactly the
response obtained by evaluating the SLS parameterization at the posterior mean:
$$
\bar\vartheta
:=
\mathbb E_{\alpha\sim\rho}[\vartheta(\alpha)]
=
\vartheta(\mu_\rho).
$$
Equivalently, denote 
$$
\bar\clmapx=\clmapx(\mu_\rho),
\quad
\bar\clmapu=\clmapu(\mu_\rho),
\quad
\bar\clmapxaff=\clmapxaff(\mu_\rho),
\quad
\bar\clmapuaff=\clmapuaff(\mu_\rho).
$$

Since the SLS achievability constraints are linear, feasibility is preserved under averaging.
The key advantage of the quadratic SLS setting is that the population risk is convex quadratic in
$\alpha$ as shown in the previous subsection.
Therefore,
$$
\mathbb E_{\alpha\sim\rho}[L(\alpha)]
=
L(\mu_\rho)+\tr(G_\alpha\Sigma_\rho).
$$
Since $G_\alpha\succeq0$, the randomization gap (or tax) satisfies
$
\operatorname{tr}(G_\alpha\Sigma_\rho)\ge0.
$
% and hence
% $$
% L(\mu_\rho)
% \le
% \mathbb E_{\alpha\sim\rho}[L(\alpha)].
% $$
Consequently, any PAC-Bayes upper bound on the randomized posterior risk also certifies the
deterministic mean response. This can also be recovered from Jensen as a one-liner~\citep{friendlyPACBayes} , but we 
want to explicitly characterize $\operatorname{tr}(G\Sigma_\rho)$ as a special important structure 
under quadratic costs.
% \begin{corollary}[Mean-response deployment certificate]
% If $G \succeq 0$ then $\operatorname{tr}(G\Sigma_\rho) \ge 0$, so $L(\mu_\rho) \le E_\rho[L(\alpha)]$. 
% Consequently, any high-probability bound $E_\rho[L(\alpha)] \le \mathcal B (\rho,S,\delta)$ also certifies $E_\rho[L(\alpha)] \le \mathcal B (\rho,S,\delta)$.
% \end{corollary}
% \begin{proof}
% Using the quadratic form of $L$,
% $$
% \mathbb E_\rho[L(\alpha)]
% =
% \mathbb E_\rho[\alpha^\top G\alpha]
% +
% 2g^\top\mathbb E_\rho[\alpha]
% +
% \alpha_0.
% $$
% Since
% $$
% \mathbb E_\rho[\alpha^\top G\alpha]
% =
% \mu_\rho^\top G\mu_\rho+\operatorname{tr}(G\Sigma_\rho),
% $$
% we obtain
% $$
% \mathbb E_\rho[L(\alpha)]
% =
% \mu_\rho^\top G\mu_\rho
% +
% 2g^\top\mu_\rho
% +
% \alpha_0
% +
% \operatorname{tr}(G\Sigma_\rho)
% =
% L(\mu_\rho)+\operatorname{tr}(G\Sigma_\rho).
% $$
% Because $G\succeq0$ and $\Sigma_\rho\succeq0$, we have
% $$
% \operatorname{tr}(G\Sigma_\rho)\ge0.
% $$
% Thus
% $$
% L(\mu_\rho)\le \mathbb E_\rho[L(\alpha)].
% $$
% The final statement follows by substitution into the PAC-Bayes bound.
% \end{proof}
The term $\operatorname{tr}(G\Sigma_\rho)$
has a useful interpretation. It is the excess risk incurred by randomized deployment relative to
mean-response deployment. It can also be seen as a curvature-weighted posterior variance term: posterior spread is
expensive in directions where the closed-loop cost has high curvature and cheap in directions where
the cost is flat.

\paragraph{Why the Dirac-posterior issue does not arise?}

At first sight, certifying the deterministic mean response may appear to conflict
with the usual PAC-Bayes requirement that the posterior have finite KL divergence
with respect to the prior. Indeed, if the prior $\pi$ is a continuous
distribution, then the Dirac measure $\delta_{\mu_\rho}$  satisfies
$
\mathrm{KL}(\delta_{\mu_\rho}\|\pi)=\infty.
$
Thus, one cannot usually apply the PAC-Bayes theorem directly to the deterministic
posterior $\delta_{\mu_\rho}$.
The mean-response certificate does not do this. Instead, the PAC-Bayes theorem is
applied to the non-degenerate posterior $\rho$, for which
$\mathrm{KL}(\rho\|\pi)<\infty$. 
%Therefore, with high probability over the
% training sample $S$, we obtain
% $$
% \mathbb E_{\alpha\sim\rho}[L(\alpha)]
% \le
% \mathcal B(\rho,S,\delta).
% $$
% Separately, because the SLS population risk is convex in the response coordinate,
% Jensen's inequality gives
% $$
% L(\mu_\rho)
% =
% L\!\left(\mathbb E_{\alpha\sim\rho}[\alpha]\right)
% \le
% \mathbb E_{\alpha\sim\rho}[L(\alpha)].
% $$
% Combining the two inequalities yields
% $$
% L(\mu_\rho)
% \le
% \mathbb E_{\alpha\sim\rho}[L(\alpha)]
% \le
% \mathcal B(\rho,S,\delta).
% $$
Hence the deterministic mean response is certified indirectly through the
randomized posterior certificate. The Dirac measure $\delta_{\mu_\rho}$ is
never used as the PAC-Bayes posterior.
In the quadratic SLS setting, this transfer is even more explicit. Since
$$
L(\alpha)=\alpha^\top G\alpha+2g^\top\alpha+\alpha_0,
\quad G\succeq0,
$$
we have the exact identity
$$
\mathbb E_{\alpha\sim\rho}[L(\alpha)]
=
L(\mu_\rho)+\operatorname{tr}(G\Sigma_\rho).
$$
% Thus,
% $$
% L(\mu_\rho)
% =
% \mathbb E_{\alpha\sim\rho}[L(\alpha)]
% -
% \operatorname{tr}(G\Sigma_\rho)
% \le
% \mathbb E_{\alpha\sim\rho}[L(\alpha)].
% $$
% The term $\operatorname{tr}(G\Sigma_\rho)$ is the randomization tax: it measures
% the excess risk of randomized deployment relative to deploying the mean response.
% The PAC-Bayes posterior is therefore used as a finite-KL certification device,
% while the actually deployed controller may be the deterministic controller
% associated with the posterior mean SLS response.

\subsection{Recovering the controller for the mean response}
\label{sec:mean-controller-equations}
The deterministic response certified by Proposition~\ref{prop:mean_response_certificate} is the mean
SLS response
$$
(\bar\clmapx,\bar\clmapu,\bar\clmapxaff,\bar\clmapuaff)
=
(\clmapx(\mu_\rho),\clmapu(\mu_\rho),\clmapxaff(\mu_\rho),\clmapuaff(\mu_\rho)).
$$
To implement this response, we recover the corresponding affine controller using the standard SLS
inversion. Since
$$
x=\bar\clmapx w+\bar\clmapxaff,
\quad
u=\bar\clmapu w+\bar\clmapuaff,
$$
and $\bar\clmapx$ is causally invertible, we have
$$
w=\bar\clmapx^{-1}(x-\bar\clmapxaff).
$$
Substituting this into the input response gives
$$
u
=
\bar\clmapu\bar\clmapx^{-1}(x-\bar\clmapxaff)+\bar\clmapuaff.
$$
Therefore, the deterministic affine controller that realizes the mean response is
$$
u=\bar K x+\bar k,
$$
where
$$
\bar K
:=
\bar\clmapu\bar\clmapx^{-1},
\quad
\bar k
:=
\bar\clmapuaff-\bar\clmapu\bar\clmapx^{-1}\bar\clmapxaff.
$$
Equivalently,
$$
\bar K
=
\clmapu(\mu_\rho)\clmapx(\mu_\rho)^{-1},
$$
and
$$
\bar k
=
\clmapuaff(\mu_\rho)
-
\clmapu(\mu_\rho)\clmapx(\mu_\rho)^{-1}\clmapxaff(\mu_\rho).
$$

Thus, the deployment procedure is to compute $\mu_\rho$, then compute, 
$
(\clmapx(\mu_\rho),\clmapu(\mu_\rho),\clmapxaff(\mu_\rho),\clmapuaff(\mu_\rho))
$
which finally give
$
(\bar K,\bar k).
$
This is not the same as averaging the feedback operators themselves. In general,
$$
\mathbb E_{\alpha\sim\rho}
\left[
\clmapu(\alpha)\clmapx(\alpha)^{-1}
\right]
\neq
\mathbb E_{\alpha\sim\rho}[\clmapu(\alpha)]
\,
\mathbb E_{\alpha\sim\rho}[\clmapx(\alpha)]^{-1}.
$$
The correct deterministic controller is obtained by first averaging in the response space and then
recovering the controller.

\section{Data-driven Gaussian certificates}
\label{app:data_driven_gaussian}
This appendix shows how the Gaussian PAC-Bayes certificate can be made data-driven when the disturbance covariance $\Sigma_w$ is unknown, hence transitioning 
from an \emph{oracle} bound to an actually computable one. 
The key idea is to estimate $\Sigma_w$ from the disturbance samples, inflate the estimate, and then use the 
inflated covariance in the Gaussian sensitivity coefficient. 
We also show that the same covariance event controls the curvature-mismatch term appearing in the deterministic mean response certificate.
Finally, our approach maintains the curvature awareness interpretation. 

For simplicity, we focus on the  zero offset setting ($m(\alpha)=0$) and zero mean disturbance, i.e.
$$
w_i \overset{\mathrm{i.i.d.}}{\sim} \mathcal N(0,\Sigma_w),
$$
for every $i \in \{1,\cdots, n\}$.
%The zero-offset assumption is used only for the curvature-mismatch refinement. T
%he covariance-inflation argument for the Gaussian sensitivity coefficient applies more generally to the covariance-dependent part of the certificate.
The loss is
$$
\ell(\alpha,w)=\|M(\alpha)w\|^2,
$$
where
$$
M(\alpha)=M_{0}+\sum_{k=1}^p \alpha_k M_{k}.
$$
Here $\alpha\in\mathbb R^p$, $w\in\mathbb R^{d_w}$, and
$
M_{k}
\in \mathbb R^{d_y\times d_w}
$
as explained in Appendix~\ref{app:nullspace_to_cost_maps} and is the consequence of the nullspace parametrization.

\subsection{Covariance event and inflated data-driven covariance}
\label{app:covariance_inflation}
We can compute the empirical covariance as
\begin{align}
\widehat \Sigma_w
=
\frac1n\sum_{i=1}^n w_iw_i^\top
\label{eq:emprical_covariance_app}
\end{align}
We assume that, for some $\epsilon_\Sigma>0$, the covariance event
\begin{align}
\mathcal E_\Sigma
=
\left\{
\|\widehat\Sigma_w-\Sigma_w\|_{\rm op}\le \epsilon_\Sigma
\right\}
\label{eq:covariance-event-app}
\end{align}
holds. We will discuss how to compute such a radius later on in Appendix~\ref{app:covariance_radius}. Instead, we now focus on the implications of this event.
% Equivalently, on $\mathcal E_\Sigma$,
% $$
% -\epsilon_\Sigma I
% \preceq
% \widehat\Sigma_w-\Sigma_w
% \preceq
% \epsilon_\Sigma I.
% $$
Covariance event given in~\eqref{eq:covariance-event-app} implies
$$
\Sigma_w
\preceq
\widehat\Sigma_w+\epsilon_\Sigma I.
$$
The inflated empirical covariance is then 
\begin{align}
\widehat\Sigma_w+\epsilon_\Sigma I.
\label{eq:inflated-covariance-app}
\end{align}
% \begin{lemma}[PSD covariance inflation]
% \label{lem:psd_covariance_inflation}
Moreover, on the event $\mathcal E_\Sigma$, for every response $\alpha$, it is not hard to see that 
\begin{align}
M(\alpha)\Sigma_wM(\alpha)^\top
\preceq
M(\alpha)(\widehat\Sigma_w+\epsilon_\Sigma I)M(\alpha)^\top.
\end{align}
Consequently 
\begin{align}
\|M(\alpha)\Sigma_wM(\alpha)^\top\|_F^2
\le
\|M(\alpha)(\widehat\Sigma_w+\epsilon_\Sigma I)M(\alpha)^\top\|_F^2.
\label{eq:inflated_sigma_y}
\end{align}
%\end{lemma}

% \begin{proof}
% The covariance event implies
% $$
% \Sigma_w\preceq \widetilde\Sigma_w.
% $$
% Premultiplying and postmultiplying by $M(\alpha)$ and $M(\alpha)^\top$ preserves the Loewner order, hence
% $$
% M(\alpha)\Sigma_wM(\alpha)^\top
% \preceq
% M(\alpha)\widetilde\Sigma_wM(\alpha)^\top.
% $$
% Both matrices are positive semidefinite. If $0\preceq A\preceq B$, then
% $$
% \|B\|_F^2-\|A\|_F^2
% =
% \operatorname{tr}(B^2)-\operatorname{tr}(A^2)
% =
% \operatorname{tr}((B-A)(B+A))
% \ge 0,
% $$
% because $B-A\succeq0$, $B+A\succeq0$, and the trace of the product of two positive semidefinite matrices is nonnegative. Applying this with
% $$
% A=M(\alpha)\Sigma_wM(\alpha)^\top,
% \quad
% B=M(\alpha)\widetilde\Sigma_wM(\alpha)^\top
% $$
% gives the claim.
% \end{proof}

\subsection{Data-driven Gaussian complexity coefficient}
\label{app:data_driven_complexity}
The Gaussian PAC-Bayes certificate involves the oracle sensitivity coefficient, which we denote
\begin{align}
C_\rho
=
\mathbb E_{\alpha\sim\rho}
\left[
\|M(\alpha)\Sigma_wM(\alpha)^\top\|_F^2
\right].
\label{eq:sensitivity-coefficient-app}
\end{align}

Since $\Sigma_w$ is unknown, $C_\rho$ is not directly computable (hence, oracle). 
On the covariance event $\mathcal E_\Sigma$, \eqref{eq:inflated_sigma_y} gives the computable upper bound
$$
C_\rho
\le
\widehat C_\rho,
$$
where the data-driven sensitivity coefficient is 
\begin{align}
\widehat C_\rho
:=
\mathbb E_{\alpha\sim\rho}
\left[
\|M(\alpha)(\widehat\Sigma_w+\epsilon_\Sigma I)M(\alpha)^\top\|_F^2
\right].
\label{eq:data-driven-sensitivity-coefficient-app}
\end{align}

Thus, the unknown covariance is handled by replacing the oracle coefficient $C_\rho$ 
with the inflated empirical coefficient $\widehat C_\rho$.
For example, if the Gaussian PAC-Bayes theorem gives a bound of the form
$$
\mathbb E_{\alpha\sim\rho}[L(\alpha)]
\le
\mathbb E_{\alpha\sim\rho}[\widehat L_S(\alpha)]
+
2\sqrt{
\frac{
C_\rho
\left(
\mathrm{KL}(\rho\|\pi)+\log(n/\delta_{\rm PB})
\right)
}{n-1}
},
$$
then, on the covariance event, it yields the data-driven bound
$$
\mathbb E_{\alpha\sim\rho}[L(\alpha)]
\le
\mathbb E_{\alpha\sim\rho}[\widehat L_S(\alpha)]
+
2\sqrt{
\frac{
\widehat C_\rho
\left(
\mathrm{KL}(\rho\|\pi)+\log(n/\delta_{\rm PB})
\right)
}{n-1}
}.
$$

It is useful to define a data-driven complexity 
$$
\mathrm{Comp}_{\rm data}(\rho)
:=
2\sqrt{
\frac{
\widehat C_\rho
\left(
\mathrm{KL}(\rho\|\pi)+\log(n/\delta_{\rm PB})
\right)
}{n-1}
}
$$
and state the following.
\begin{proposition}[Data-driven randomized Gaussian certificate]
\label{prop:data_driven_randomized_gaussian}
Suppose the PAC-Bayes event holds with a probability of at least $1-\delta_{\rm PB}$ and the covariance event $\mathcal E_\Sigma$ holds. Let $C_\rho$ and $\widehat C_\rho$ 
be as in $\eqref{eq:sensitivity-coefficient-app},\eqref{eq:data-driven-sensitivity-coefficient-app}$.
Then, on the intersection of these events, for all posteriors $\rho \ll \pi$,
\begin{align}
\mathbb E_{\alpha\sim\rho}[L(\alpha)]
\le
\mathbb E_{\alpha\sim\rho}[\widehat L_S(\alpha)]
+
\mathrm{Comp}_{\rm data}(\rho).
\label{eq:data-driven-bound-app}
\end{align}
If $\mathbb P(\mathcal E_\Sigma)\ge 1-\delta_\Sigma$, then the bound~\eqref{eq:data-driven-bound-app} holds with a probability of at least $1-\delta_{\rm PB}-\delta_\Sigma$.
\end{proposition}

\begin{proof}
The oracle Gaussian PAC-Bayes theorem gives, on the PAC-Bayes event,
$$
\mathbb E_{\alpha\sim\rho}[L(\alpha)]
\le
\mathbb E_{\alpha\sim\rho}[\widehat L_S(\alpha)]
+
2\sqrt{
\frac{
C_\rho
\left(
\mathrm{KL}(\rho\|\pi)+\log(n/\delta_{\rm PB})
\right)
}{n-1}
}.
$$
On $\mathcal E_\Sigma$, it holds that $C_\rho\le \widehat C_\rho$. Substitution gives the stated bound. The probability statement follows by the union bound.
\end{proof}

\subsection{Curvature matrices and mean-response decomposition}
\label{app:curvature_matrices}

For fixed $w$, define the sample-dependent design matrix in the response coordinate $\alpha$ by
$$
\mathcal A(w)
:=
\begin{bmatrix}
M_{1}w & \cdots & M_{p}w
\end{bmatrix}
\in\mathbb R^{d_y\times p}.
$$
Then
$$
M(\alpha)w
=
M_{0}w+\mathcal A(w)\alpha.
$$
Let
$$
a_0(w):=M_{0}w.
$$
The loss can be written as
$$
\ell(\alpha,w)
=
\|a_0(w)+\mathcal A(w)\alpha\|^2.
$$
Expanding in $\alpha$ gives
$$
\ell(\alpha,w)
=
\alpha^\top \mathcal A(w)^\top\mathcal A(w)\alpha
+
2a_0(w)^\top\mathcal A(w)\alpha,
+
\|a_0(w)\|^2.
$$

From previous sections (see Appendix~\ref{app:quadraticity_alpha}), we know that 
$$
G_\alpha
:=
\mathbb E\left[\mathcal A(w)^\top\mathcal A(w)\right],
%\in\mathbb R^{p\times p}, 
\quad g_\alpha
:=
\mathbb E\left[\mathcal A(w)^\top a_0(w)\right],
\quad 
c_\alpha
:=
\mathbb E\left[\|a_0(w)\|^2\right]
$$
and we can define the empirical counterparts as 
$$
\widehat G_S
:=
\frac1n\sum_{i=1}^n
\mathcal A(w_i)^\top\mathcal A(w_i),
%\in\mathbb R^{p\times p}, 
\quad \widehat g_S
:=
\frac1n\sum_{i=1}^n
\mathcal A(w_i)^\top a_0(w_i),
\quad
\widehat c_S
:=
\frac1n\sum_{i=1}^n
\|a_0(w_i)\|^2.
$$
% Similarly define
% $$
% g
% :=
% \mathbb E\left[\mathcal A(w)^\top a_0(w)\right],
% \quad
% \widehat g_S
% :=
% \frac1n\sum_{i=1}^n
% \mathcal A(w_i)^\top y_0(w_i),
% $$
% and
% $$
% c
% :=
% \mathbb E\left[\|a_0(w)\|^2\right],
% \quad
% \widehat c_S
% :=
% \frac1n\sum_{i=1}^n
% \|y_0(w_i)\|^2.
% $$
Then
$$
L(\alpha)
=
\alpha^\top G_\alpha\alpha+2g_\alpha^\top\alpha+c_\alpha,
$$
and
$$
\widehat L_S(\alpha)
=
\alpha^\top \widehat G_S\alpha
+
2\widehat g_S^\top\alpha
+
\widehat c_S.
$$

Let
$$
\mu_\rho:=\mathbb E_{\alpha\sim\rho}[\alpha],
\quad
\Sigma_\rho
:=
\mathbb E_{\alpha\sim\rho}
\left[
(\alpha-\mu_\rho)(\alpha-\mu_\rho)^\top
\right].
$$
% For any quadratic function
% $$
% q(\alpha)=\alpha^\top H\alpha+2h^\top\alpha+c,
% $$
% we have
% $$
% \mathbb E_{\alpha\sim\rho}[q(\alpha)]
% =
% q(\mu_\rho)+\operatorname{tr}(H\Sigma_\rho).
% $$
In light of Proposition~\ref{prop:quadratic_decomposition}%, $L$ and $\widehat L_S$ are
$$
\mathbb E_{\alpha\sim\rho}[L(\alpha)]
=
L(\mu_\rho)+\operatorname{tr}(G\Sigma_\rho),
$$
and
$$
\mathbb E_{\alpha\sim\rho}[\widehat L_S(\alpha)]
=
\widehat L_S(\mu_\rho)+\operatorname{tr}(\widehat G_S\Sigma_\rho).
$$

Combining these identities with Proposition~\ref{prop:data_driven_randomized_gaussian} yields
$$
L(\mu_\rho)
\le
\widehat L_S(\mu_\rho)
+
\mathrm{Comp}_{\rm data}(\rho)
+
\operatorname{tr}\left((\widehat G_S-G)\Sigma_\rho\right).
$$
The remaining task is to somehow control the curvature mismatch
$$
\operatorname{tr}\left((\widehat G_S-G)\Sigma_\rho\right).
$$

\subsection{Curvature mismatch as a covariance-estimation error}
\label{app:curvature_mismatch}
We now show that the curvature mismatch can be controlled using the same covariance event $\mathcal E_\Sigma$.
It will be useful to define
$$
r_\rho(w)
:=
\operatorname{tr}\left(
\mathcal A(w)\Sigma_\rho\mathcal A(w)^\top
\right).
$$
By the cyclicity of the trace,
$$
r_\rho(w)
=
\operatorname{tr}\left(
\mathcal A(w)^\top\mathcal A(w)\Sigma_\rho
\right).
$$
Therefore,
$$
\operatorname{tr}(\widehat G_S\Sigma_\rho)
=
\frac1n\sum_{i=1}^n r_\rho(w_i),
$$
and
$$
\operatorname{tr}(G\Sigma_\rho)
=
\mathbb E[r_\rho(w)].
$$
Hence
$$
\operatorname{tr}\left((\widehat G_S-G)\Sigma_\rho\right)
=
\frac1n\sum_{i=1}^n r_\rho(w_i)
-
\mathbb E[r_\rho(w)].
$$

Next, we can expand $r_\rho(w)$ since
$$
\mathcal A(w)
=
\begin{bmatrix}
M_{1}w & \cdots & M_{p}w
\end{bmatrix}.
$$
We have
$$
r_\rho(w)
=
\sum_{i,j=1}^p
(\Sigma_\rho)_{ij}
(M_{i}w)^\top(M_{j}w).
$$
Rewriting
$
(M_{i}w)^\top(M_{j}w)
=
w^\top (M_{i})^\top M_{j}w,
$
gives
$
r_\rho(w)
=
w^\top A_\rho w,
$
where
\begin{align}
A_\rho
:=
\sum_{i,j=1}^p
(\Sigma_\rho)_{ij}
(M_{i})^\top M_{j}
\in \mathbb R^{d_w\times d_w}.
\label{eq:Arho}
\end{align}
Moreover $A_\rho\succeq0$, because
$$
w^\top A_\rho w
=
\operatorname{tr}\left(
\mathcal A(w)\Sigma_\rho\mathcal A(w)^\top
\right)
=
\|\mathcal A(w)\Sigma_\rho^{1/2}\|_F^2
\ge 0
$$
for all $w$.

Using $w^\top A_\rho w=\operatorname{tr}(A_\rho ww^\top)$, we get
$$
\frac1n\sum_{i=1}^n r_\rho(w_i)
=
\operatorname{tr}(A_\rho\widehat\Sigma_w),
$$
and
$$
\mathbb E[r_\rho(w)]
=
\operatorname{tr}(A_\rho\Sigma_w).
$$
Therefore, we obtain the key identity
\begin{align}
\operatorname{tr}\left((\widehat G_S-G)\Sigma_\rho\right)
=
\operatorname{tr}\left(A_\rho(\widehat\Sigma_w-\Sigma_w)\right).
\label{eq:key-identity-Arho}
\end{align}
\begin{lemma}[Curvature mismatch under covariance inflation]
\label{lem:curvature_mismatch_covariance}
On the covariance event $\mathcal E_\Sigma$,
$$
\operatorname{tr}\left((\widehat G_S-G)\Sigma_\rho\right)
\le
\epsilon_\Sigma \operatorname{tr}(A_\rho).
$$
\end{lemma}

\begin{proof}
On $\mathcal E_\Sigma$, it holds 
$
\widehat\Sigma_w-\Sigma_w
\preceq
\epsilon_\Sigma I.
$
Since $A_\rho\succeq0$,
$$
\operatorname{tr}\left(A_\rho(\widehat\Sigma_w-\Sigma_w)\right)
\le
\operatorname{tr}\left(A_\rho(\epsilon_\Sigma I)\right)
=
\epsilon_\Sigma\operatorname{tr}(A_\rho).
$$
Using the identity~\eqref{eq:key-identity-Arho} gives the result.
\end{proof}

For diagonal posterior covariance,
$$
\Sigma_\rho
=
\operatorname{diag}(\sigma_1^2,\ldots,\sigma_p^2),
$$
the matrix $A_\rho$ simplifies to
$$
A_\rho
=
\sum_{k=1}^p
\sigma_k^2
(M_{k})^\top M_{k}.
$$
Consequently,
$$
\operatorname{tr}(A_\rho)
=
\sum_{k=1}^p
\sigma_k^2
\|M_{k}\|_F^2.
$$
Thus, the mismatch penalty becomes the computable weighted posterior-variance penalty
$$
\epsilon_\Sigma\operatorname{tr}(A_\rho)
=
\epsilon_\Sigma
\sum_{k=1}^p
\sigma_k^2
\|M_{k}\|_F^2.
$$

\subsection{ Final data-driven deterministic certificate}
\label{app:final_data_driven_certificate}
Finally, we now combine the data-driven randomized PAC-Bayes certificate with the curvature-mismatch bound to round off the story via the following proposition.

\begin{proposition}[Data-driven deterministic certificate]
\label{prop:data-driven-derand}
Assume the zero mean, zero offset Gaussian setting
$w_i \overset{\text{i.i.d.}}{\sim} \mathcal{N}(0,\Sigma_w)$ with $m(\alpha)=0$.
Suppose the PAC-Bayes event holds with probability at least $1-\delta_{\mathrm{PB}}$
and the covariance event $\mathcal{E}_\Sigma$ holds with probability at least
$1-\delta_\Sigma$. Let $\rho \ll \pi$ be a posterior with
$\mathbb{E}_\rho[\|\alpha\|^2] < \infty$, mean $\mu_\rho$ and covariance $\Sigma_\rho$.
Then, with probability at least $1-\delta_{\mathrm{PB}}-\delta_\Sigma$,
$$
L(\mu_\rho) \;\le\; \widehat{L}_S(\mu_\rho)
+ \mathrm{Comp}_{\mathrm{data}}(\rho)
+ \epsilon_\Sigma \,\mathrm{tr}(A_\rho),
$$
where $A_\rho := \sum_{i,j=1}^{p} (\Sigma_\rho)_{ij}\, M_i^\top M_j$ as in \eqref{eq:Arho}.
Equivalently, using $\mathrm{tr}(A_\rho)=\mathrm{tr}(\Sigma_\rho M_{\mathrm{G}})$
with the Gramian $(M_{\mathrm{G}})_{ij}=\mathrm{tr}(M_i^\top M_j)$, the mismatch term
recovers the form of Proposition~\eqref{prop:data_driven_mean_certificate}. 
For a diagonal posterior covariance
$\Sigma_\rho = \mathrm{diag}(\sigma_1^2,\dots,\sigma_p^2)$,
\begin{align}
L(\mu_\rho) \;\le\; \widehat{L}_S(\mu_\rho)
+ \mathrm{Comp}_{\mathrm{data}}(\rho)
+ \epsilon_\Sigma \sum_{k=1}^{p} \sigma_k^2 \,\|M_k\|_F^2 .
\end{align}
\end{proposition}

\begin{proof}
On the intersection of the PAC-Bayes event and $\mathcal{E}_\Sigma$,
Proposition~\ref{prop:data_driven_randomized_gaussian} gives
\begin{align}
\mathbb{E}_{\alpha\sim\rho}[L(\alpha)]
\;\le\; \mathbb{E}_{\alpha\sim\rho}[\widehat{L}_S(\alpha)]
+ \mathrm{Comp}_{\mathrm{data}}(\rho).
\end{align}
Since $\mathbb{E}_\rho[\|\alpha\|^2]<\infty$, the second-moment identities of
Proposition~\ref{prop:quadratic_decomposition} apply, and with
$G_\alpha = \mathbb{E}[A(w)^\top A(w)]$ and its empirical counterpart $\widehat{G}_S$,
\begin{align}
\mathbb{E}_{\alpha\sim\rho}[L(\alpha)] = L(\mu_\rho) + \mathrm{tr}(G_\alpha \Sigma_\rho),
\quad
\mathbb{E}_{\alpha\sim\rho}[\widehat{L}_S(\alpha)] = \widehat{L}_S(\mu_\rho) + \mathrm{tr}(\widehat{G}_S \Sigma_\rho).
\end{align}
Substituting both identities and rearranging yields
$$
L(\mu_\rho) \;\le\; \widehat{L}_S(\mu_\rho)
+ \mathrm{Comp}_{\mathrm{data}}(\rho)
+ \mathrm{tr}\!\big((\widehat{G}_S - G_\alpha)\Sigma_\rho\big).
$$
By the key identity~\eqref{eq:key-identity-Arho},
$\mathrm{tr}\big((\widehat{G}_S - G_\alpha)\Sigma_\rho\big)
= \mathrm{tr}\big(A_\rho(\widehat{\Sigma}_w - \Sigma_w)\big)$,
 Lemma~\ref{lem:curvature_mismatch_covariance} applies and gives
\begin{align*}
\mathrm{tr}\!\big((\widehat{G}_S - G_\alpha)\Sigma_\rho\big) \;\le\; \epsilon_\Sigma \,\mathrm{tr}(A_\rho)
\quad \text{on } \mathcal{E}_\Sigma.
\end{align*}
Combining the two inequalities gives the first claim. For diagonal $\Sigma_\rho$,
$A_\rho = \sum_{k=1}^{p}\sigma_k^2 M_k^\top M_k$, hence
$\mathrm{tr}(A_\rho) = \sum_{k=1}^{p}\sigma_k^2 \|M_k\|_F^2$, which yields the
diagonal posterior formula. The probability statement follows by the union bound.
\end{proof}

\subsection{Explicit Gaussian covariance radius}
\label{app:covariance_radius}

The preceding results are stated conditionally on a valid covariance event. We now give one explicit choice of $\epsilon_\Sigma$ for Gaussian disturbances.

Assume
$$
w_i \overset{\mathrm{i.i.d.}}{\sim}
\mathcal N(0,\Sigma_w),
\quad
\|\Sigma_w\|_{\rm op}\le \bar\sigma_w^2.
$$
Let
$$
\widehat\Sigma_w
=
\frac1n\sum_{i=1}^n w_iw_i^\top.
$$
Fix $\delta_\Sigma\in(0,1)$ and define
$$
a_{n,\delta}
:=
\sqrt{\frac{d_w}{n}}
+
\sqrt{\frac{2\log(2/\delta_\Sigma)}{n}}.
$$
A standard Gaussian sample covariance concentration bound gives, with probability at least $1-\delta_\Sigma$,
$$
\|\widehat\Sigma_w-\Sigma_w\|_{\rm op}
\le
\|\Sigma_w\|_{\rm op}
\left(
2a_{n,\delta}+a_{n,\delta}^2
\right).
$$
Therefore, using the scale bound $\|\Sigma_w\|_{\rm op}\le \bar\sigma_w^2$, one may take
$$
\epsilon_\Sigma
=
\bar\sigma_w^2
\left(
2a_{n,\delta}+a_{n,\delta}^2
\right).
$$

This explicit choice is conservative and depends on the ambient disturbance dimension $d_w$. Sharper covariance concentration inequalities, for example effective-rank bounds, can be substituted without changing the PSD-inflation and curvature-mismatch arguments above. The certificate only requires a valid event of the form
$$
\|\widehat\Sigma_w-\Sigma_w\|_{\rm op}\le \epsilon_\Sigma.
$$

% \subsection{Proof of mismatch}
% \label{app-sec:proof-covariance_mismatch}
Next we prove that this radius is a sound choice.
\begin{proof}
% Let $r=\operatorname{rank}(\Sigma_w)\le d_w$.  Note that $r$ is generally not known, but we will come back to that later.
% On the range of $\Sigma_w$,
% write
% $$
% w_i=\Sigma_w^{1/2}z_i,
% \qquad
% z_i\sim\mathcal N(0,I_r).
% $$
% Let $Z\in\mathbb R^{n\times r}$ have rows $z_i^\top$. Then
% $$
% \widehat\Sigma_w-\Sigma_w
% =
% \Sigma_w^{1/2}
% \left(
% \frac1n Z^\top Z-I_r
% \right)
% \Sigma_w^{1/2},
% $$
% and therefore
% $$
% \|\widehat\Sigma_w-\Sigma_w\|_{\rm op}
% \le
% \|\Sigma_w\|_{\rm op}
% \left\|
% \frac1n Z^\top Z-I_r
% \right\|_{\rm op}.
% $$
Let $r:=\operatorname{rank}(\Sigma_w)\le d_w$. Write the reduced
spectral decomposition
$$
\Sigma_w = Q_R \Lambda_R Q_R^\top,
$$
where $Q_R\in\mathbb{R}^{d_w\times r}$ has orthonormal columns and
$\Lambda_R\in\mathbb{R}^{r\times r}$ is positive definite. Define 
$$
S:=Q_R\Lambda_R^{1/2}\in\mathbb{R}^{d_w\times r},
$$
so that $\Sigma_w=SS^\top$.
Since $w_i\sim\mathcal N(0,\Sigma_w)$, we may write
$$
w_i = S z_i,
\qquad
z_i\sim\mathcal N(0,I_r).
$$
Let $Z\in\mathbb{R}^{n\times r}$ be the matrix whose $i$-th row is
$z_i^\top$. Then
$$
\widehat\Sigma_w
=
\frac{1}{n}\sum_{i=1}^n w_iw_i^\top
=
S\left(\frac1n Z^\top Z\right)S^\top,
$$
and hence
$$
\widehat\Sigma_w-\Sigma_w
=
S\left(\frac1n Z^\top Z-I_r\right)S^\top.
$$
Therefore,
$$
\begin{aligned}
\|\widehat\Sigma_w-\Sigma_w\|_{\mathrm{op}}
&\leq
\|S\|_{\mathrm{op}}^2
\left\|
\frac1n Z^\top Z-I_r
\right\|_{\mathrm{op}} \\
&=
\|\Sigma_w\|_{\mathrm{op}}
\left\|
\frac1n Z^\top Z-I_r
\right\|_{\mathrm{op}}.
\end{aligned}
$$

Let $s_j(Z)$ be the $j$-th singular value of the matrix $Z$, and let the minimum and maximum values be denoted by $s_{\min}$ and $s_{\max}$. 
By the Gaussian extreme singular-value inequality~\citep{davidson2001local}\citep[Theorem~2.6,eq~(2.3)]{rudelson2010non-singular_values}, with probability
at least 
$
1-2\exp{(-\tfrac{t^2}{2})}
$ for $t \ge 0$
\begin{align}
s_{\max}(Z)\le \sqrt n+\sqrt r+t
\end{align}
and
\begin{align}
s_{\min}(Z)\ge \sqrt n-\sqrt r-t.
\end{align}
Setting $t =\sqrt{2\log(2/\delta_\Sigma)}$ gives the probability of $1 - \delta_\Sigma $ and
\begin{align}
s_{\max}(Z)\le \sqrt n+\sqrt r+\sqrt{2\log(2/\delta_\Sigma)}
\end{align}
and
\begin{align}
s_{\min}(Z)\ge \sqrt n-\sqrt r-\sqrt{2\log(2/\delta_\Sigma)}.
\end{align}
Divide both inequalities by $\sqrt{n}$ and define
$$
a(n,\delta_\Sigma,r)
=
\sqrt{\frac{r}{n}}
+
\sqrt{\frac{2\log(2/\delta_\Sigma)}{n}},
$$
we have
$$
\left\|
\frac1n Z^\top Z-I_r
\right\|_{\rm op}
\le
\max\left\{
(1+a(n,\delta_\Sigma,r))^2-1,\,
1-[1-a(n,\delta_\Sigma,r)]_+^2
\right\}.
$$
The clamping to one is necessary because if $a>1$ the $Z$ is rank deficient and has zero singular values, but 
due to the identity matrix, the overall matrix must have a singular value of one.
Since
$$
(1+a)^2-1=2a+a^2
$$
and
$$
1-[1-a]_+^2\le 2a+a^2,
$$
we combine both into
$$
\left\|
\frac1n Z^\top Z-I_r
\right\|_{\rm op}
\le
2a(n,\delta_\Sigma,r)+(a(n,\delta_\Sigma,r))^2.
$$
Finally, a safe (and conservative) choice for $r$ is $r =  d_w$, so $a(n,\delta_\Sigma,r)\le a(n,\delta_\Sigma,d_w)$, and
$\|\Sigma_w\|_{\rm op}\le\bar\sigma_w^2$. This proves the claim.
\end{proof}

\subsection{Remark on affine offsets}
\label{app:affine_offsets_data_driven}

The clean curvature-mismatch formula above assumes $m(\alpha)=0$. If affine offsets are present, then
$$
M(\alpha)w+m(\alpha)
=
\bar M(\alpha)\bar w,
\quad
\bar w:=
\begin{bmatrix}
w\\
1
\end{bmatrix},
$$
where
$$
\bar M(\alpha)
=
\begin{bmatrix}
M(\alpha) & m(\alpha)
\end{bmatrix}.
$$
The same algebra can be repeated with the augmented second-moment matrix
$$
\mathbb E[\bar w\bar w^\top]
=
\begin{bmatrix}
\Sigma_w & 0\\
0 & 1
\end{bmatrix}
$$
in the zero-mean case. However, if the mean is unknown or the offset terms are estimated from data, 
then an additional concentration event for the mean or the augmented second moment is needed. 
For this reason, the main data-driven curvature-mismatch refinement is stated in the zero-mean, zero-offset setting.

%\input{Appendix/distro_inversion}
%\input{Appendix/app_extended_experiments}
%\include{Appendix/app_deterministic_deployment}
%\input{Appendix/quadratic_SLS}
%\include{Appendix/appendix_why_not_K}
%\include{Appendix/appendix_experiments}
%\section{Old}
%\include{old/storage}
%%%%%%%%%%%%%%%%%%%%%%%%%%%%%%%%%%%%%%%%%%%%%%%%%%%%%%%%%%%%

%TODO: DON'T FORGET TO INCLUDE CHECKLIST>
\newpage
%\input{checklist.tex} %INCLUDE CHECKLIST!!!!!!!!!!!!!!!!!!!!!!!!!!!!!!!

    % uses references.bib
\end{document}